\documentclass[prd]{revtex4}%
\usepackage{amsfonts}
\usepackage{amsmath}
\usepackage{amssymb}
\usepackage{graphicx}%
\providecommand{\U}[1]{\protect\rule{.1in}{.1in}}
\newtheorem{theorem}{Theorem}

\newenvironment{proof}[1][Proof]{\noindent\textbf{#1.} }{\ \rule{0.5em}{0.5em}}
\begin{document}
\title[Bianchi \textrm{VI}$_{0}$ in Scalar and Scalar-Tensor Cosmologies]{Bianchi \textrm{VI}$_{0}$ in Scalar and Scalar-Tensor cosmologies}
\author{J.A. Belinch\'{o}n}
\affiliation{Departamento de F\'{\i}sica At\'{o}mica, Molecular y
Nuclear. Ciencias F\'{\i}sicas. Universidad Complutense de Madrid,
E-28040 Madrid, Espa\~{n}a} \keywords{Scalar and Scalar-Tensor
models, Bianchi VI$_{0}$, Self-similar solutions, Time varying
constants} \pacs{PACS number}

\begin{abstract}
We study several cosmological models with Bianchi \textrm{VI}$_{0}$ symmetries
under the self-similar approach. In order to study how the \textquotedblleft
constants\textquotedblright\ $G$ and $\Lambda$ may vary, we propose three
scenarios where such constants are considered as time functions. The first
model is a perfect fluid. We find that the behavior of $G$ and $\Lambda$ are
related. If $G$ behaves as a growing time function then $\Lambda$ is a
positive decreasing time function but if $G$ is decreasing then $\Lambda$ is
negative. For this model we have found a new solution. The second model is a
scalar field, where in a phenomenological way, we consider a modification of
the Klein-Gordon equation in order to take into account the variation of $G$.
Our third scenario is a scalar-tensor model. We find three solutions for this
models where $G$ is growing, constant or decreasing and $\Lambda$ is a
positive decreasing function or vanishes. We put special emphasis on
calculating the curvature invariants in order to see if the solutions isotropize.

\end{abstract}
\date{\today}
\maketitle

\section{Introduction}

Current observations of the large scale cosmic microwave background (CMB)
suggest to us that our physical universe is expanding in an accelerated way,
isotropic and homogeneous models with a positive cosmological constant. The
analysis of CMB fluctuations could confirm this picture. But other analyses
reveal some inconsistencies. Analysis of WMAP data sets shows us that the
universe might have a preferred direction. For this reason, it may be
interesting to study Bianchi models since these models may describe such anisotropies.

The observed location of the first acoustic peak of the temperature
fluctuations on the CMB corroborated by the data obtained in different
experiments \cite{bernadis}, indicates that the universe is dominated by an
unidentified \textquotedblleft dark energy\textquotedblright\ and suggests
that this unidentified dark energy has a negative pressure \cite{perlmutter}.
This last characteristic of the dark energy points to the vacuum energy or
cosmological constant as a possible candidate for dark energy.

In modern cosmological theories, the cosmological constant remains a focal
point of interest (see \cite{cc1}-\cite{cc4}\ \ for reviews of the problem). A
wide range of observations now compellingly suggest that the universe
possesses a non-zero cosmological constant. Some of the recent discussions on
the cosmological constant \textquotedblleft problem\textquotedblright\ and on
cosmology with a time-varying cosmological constant point out that in the
absence of any interaction with matter or radiation, the cosmological constant
remains a \textquotedblleft constant\textquotedblright. However, in the
presence of interactions with matter or radiation, a solution of Einstein
equations and the assumed equation of covariant conservation of stress-energy
with a time-varying $\Lambda$ can be found. This entails that energy has to be
conserved by a decrease in the energy density of the vacuum component followed
by a corresponding increase in the energy density of matter or
radiation.\ Recent observations strongly favour a significant and a positive
value of $\Lambda$ with magnitude $\Lambda(G\hbar/c^{3})\approx10^{-123}$.
These observations suggest an accelerating expansion of the universe, $q<0,$
\cite{perlmutter}.

Following Maia, et al \cite{Mia} who have pointed out that although the
cosmological $\Lambda$-term has a very small value today, it may contribute to
the total energy density of the universe. For this reason, since its present
value, $\Lambda_{0}$, may be a remnant of a primordial inflationary stage, it
seems natural to study cosmological scenarios which include a decaying vacuum
energy density in such a way that it must be high enough at very early times
and sufficiently small at present times in order to be compatible with the
current observations. One of the first attempts at considering a decreasing
cosmological term was formulated by Chen et al \cite{CW}. By studying the
Wheeler-DeWitt equation, they argue through dimensional considerations that
the cosmological $\Lambda$-term must follow a relationship such as
$\Lambda\sim t^{-2},$ in order to fit with current observations. Other
mechanism to describe such variation have been formulated within the framework
of the so-called \textquotedblleft quintessence models\textquotedblright.
Recently this class of models have received a great deal of attention
\cite{quintessence} and \cite{steinhardt}. Taking into account different
observational data it is possible to rule out and to obtain \textquotedblleft
correct\textquotedblright\ potential which could play the role of an effective
cosmological constant. This strengthens the idea of considering alternative
theories where the scalar field is non-minimally coupled to gravity, like
scalar-tensor theories (STT) \cite{Pimen}. This class of theories furthermore
allows the variation of other constants such as the Newton gravitational one.
There are several STT derived from the original one, the Brans-Dicke (BD)
model (see for example \cite{JBD1}-\cite{bertolami}). They have been
formulated as possible solutions to the discrepancies with observations and
try to explain the behaviour of the universe at late times (see \cite{sen}%
-\cite{bartolo}). Of particular interest are the so called chameleon
scalar-tensor theories \cite{Chame}.

The study of self-similar (SS) models is quite important since a large class
of orthogonal spatially homogeneous models are asymptotically self-similar at
the initial singularity and are approximated by exact perfect fluid or vacuum
self-similar power-law models. Exact self-similar power-law models can also
approximate general Bianchi models at intermediate stages of their evolution.
This last point is of particular importance in relating Bianchi models to the
real Universe. At the same time, self-similar solutions can describe the
behaviour of Bianchi models at late times i.e. as $t\rightarrow\infty$ (see
\cite{ColeyDS}).

The aim of this work is to study self-similar solutions of a Bianchi
$\mathrm{VI}_{0}$ cosmological model in different contexts and where the
\textquotedblleft constants\textquotedblright\ $G$ and $\Lambda$ may vary. We
are mainly interested in finding exact solutions for the proposed models as
well as to compare the behavior of $G$ and $\Lambda$ in the different
contexts. In section 2 we start showing all the geometrical ingredients that
we are going to use throughout the paper. We put special emphasis on the study
of the curvature invariants in order to study whether the obtained solutions
isotropize. Once we have calculated the homothetic vector field (HVF) in
section 3 we study the \textquotedblleft classical\textquotedblright\ solution
for a vacuum and perfect fluid models comparing these solutions with those
obtained ones in a previous work (in that work we used another Bianchi
$\mathrm{VI}_{0}$ metric \cite{Tony1}) as well as a perfect fluid model with
time-varying constants. In section 4, we start by studying the kind of
potential and scalar fields compatible with self-similar solution. The stated
theorems are very general and are valid for all the Bianchi models as well as
for the FRW models. All the proofs have been performed by studying the
Klein-Gordon equation through the Lie group method. Once we know the scalar
fields compatible with the self-similar solution we continue studying a simple
scalar model as well as a non-interacting scalar model with matter. In order
to incorporate a gravitational \textquotedblleft
varying-constant\textquotedblright\ $G(t)$ within this framework we purpose,
in a phenomenological way, a modified Klein-Gordon equation. As above we need
to study the class of potential compatible with a self-similar solution and a
varying $G$. Two kinds of models are studied. In section 5 we study a
generalized scalar-tensor model that determine an accelerated expansion at the
present epoch, with arbitrary $\omega(\phi)=const.$ and $\Lambda(\phi)$, where
this last function plays the role of an effective cosmological constant. We
would like to emphasize that in order to study the resulting field equations
(FE) we have not needed to make any assumption, otherwise, we have deduced,
the form of $\Lambda(\phi)$ by studying the conservation equation through the
Lie group method. In section 6, we summarize our results. Finally, in the
appendix A, we study through the matter collineation method the kinds of
potentials compatible with a self-similar solution in the framework of $G$
constant. In appendix B we study, using the same method, the $G-$var framework
in such a way that we regain through this method the results obtained in
section 4.

\section{The geometric ingredients}

We start by considering the following Killing vector fields (KFV) (see
\cite{ES})
\begin{equation}
\xi_{1}=\partial_{x}+mz\partial_{y}+my\partial_{z},\qquad\xi_{2}=\partial
_{y},\qquad\xi_{3}=\partial_{z}, \label{B6_K}%
\end{equation}
then%
\begin{equation}
\left[  \xi_{1},\xi_{2}\right]  =-m\xi_{3},\qquad\left[  \xi_{2},\xi
_{3}\right]  =0,\qquad\left[  \xi_{3},\xi_{1}\right]  =m\xi_{2}.
\end{equation}
Note that in this approach is essential to consider the $m-$parameter,
otherwise it is impossible to obtain self-similar (SS) solutions.

In this way it is obtained the following vector fields $\left\{
X_{j}\right\}  $, such that, $\left[  \xi_{i},X_{j}\right]  =0,$ $\left[
X_{i},X_{j}\right]  =-C_{ij}^{k}X_{k}:$%
\begin{equation}
X_{1}=\cosh mx\partial_{y}+\sinh mx\partial_{z},\;X_{2}=\sinh mx\partial
_{y}+\cosh mx\partial_{y},\;X_{3}=\partial_{x},
\end{equation}
and the dual 1-forms:%
\begin{equation}
\omega^{1}=\mathrm{dx,\qquad}\omega^{2}=\cosh mx\mathrm{dy}-\sinh
mx\mathrm{dz,\qquad}\omega^{3}=-\sinh mx\mathrm{dy}+\cosh mx\mathrm{dz,}
\label{4}%
\end{equation}

The metric is defined by
\begin{equation}
ds^{2}=-c^{2}dt^{2}+a^{2}(t)\left(  \omega^{1}\right)  ^{2}+b^{2}(t)\left(
\omega^{2}\right)  ^{2}+d^{2}(t)\left(  \omega^{3}\right)  ^{2},
\end{equation}
finding that following metric when using Eq. (\ref{4})%
\begin{align}
ds^{2}  &  =-\mathrm{dt}^{2}+a^{2}\mathrm{dx}^{2}+\left(  b^{2}\cosh
^{2}mx+d^{2}\sinh^{2}mx\right)  \mathrm{dy}^{2}\nonumber\\
&  -2\left(  b^{2}+d^{2}\right)  \cosh mx\sinh mx\mathrm{dydz}+\left(
b^{2}\sinh^{2}mx+d^{2}\cosh^{2}mx\right)  \mathrm{dz}^{2}, \label{BVIo-metric}%
\end{align}
where we have set $c=1.$

We may define the four velocity as follows: $u^{i}=\left(  1,0,0,0\right)  ,$
in such a way that it is verified, $g(u^{i},u^{i})=-1.$ From the definition of
the $4-$velocity we find that:
\begin{equation}
H=\frac{1}{3}\left(  \frac{a^{\prime}}{a}+\frac{b^{\prime}}{b}+\frac
{d^{\prime}}{d}\right)  =\frac{1}{3}\sum_{i}H_{i},\quad q=\frac{d}{dt}\left(
\frac{3}{H}\right)  -1,
\end{equation}
and
\begin{equation}
\sigma^{2}=\frac{1}{3}\left(  \sum_{i}H_{i}^{2}-\sum_{i\neq j}H_{i}%
H_{j}\right)  .
\end{equation}

Isotropization means, in essence, that at large physical times $t$, when the
volume factor, $v=abd,$ tends to infinity, the three scale factors $\left(
a,b,d\right)  $ grow at the same rate \cite{KB}. We will therefore say, by
definition, that a model is isotropizing if, for each scale factors,
$a/f\rightarrow const>0,$ $b/f\rightarrow const>0$ and $d/f\rightarrow
const>0,$ as $v\rightarrow\infty,$ where, $f=v^{1/3},$ is the mean scale
factor. Then, by rescaling some of the coordinates, we can make
$a/f\rightarrow1$, $b/f\rightarrow1,$ $d/f\rightarrow1$ and the metric will
become manifestly isotropic at large $t$. Two such criteria are $\mathcal{A}%
\rightarrow0$ and $\sigma\rightarrow0$, where the mean anisotropy parameter
$\mathcal{A}$ is defined for the metric as (see, e.g., \cite{H})%

\begin{equation}
\mathcal{A=}\frac{\sigma^{2}}{3H^{2}}=\frac{1}{3}\sum\frac{H_{i}^{2}}{H^{2}%
}-1. \label{MA}%
\end{equation}

The mean anisotropy parameter gives a dimensionless measure of the anisotropy
in the Hubble flow by comparing the shear scalar $\sigma$ to the overall rate
of expansion as described by $H$. The anisotropy in the temperature of the
CMBR enables one to estimate the value of $\sigma^{2}$ at the present epoch.

We also study the curvature behaviour of the solutions (see for example
\cite{Caminati, gron1, gron2} and \cite{Barrow}). The studied curvature
quantities are the following ones: Ricci scalar, $I_{0}=R_{i}^{i},$
Krestchmann scalar, $I_{1}=R_{ijkl}R^{ijkl}$, the full contraction of the
Ricci tensor, $I_{2}=R_{ij}R^{ij}.$ The Weyl scalar, $I_{3}=C^{abcd}%
C_{abcd}=I_{1}-2I_{2}+\frac{1}{3}I_{0}^{2},$ as well as the electric scalar
$I_{4}=E_{ij}E^{ij},$ \cite{Coley} and the magnetic scalar $I_{5}=H_{ij}%
H^{ij},$ of the Weyl tensor. The Weyl parameter $\mathcal{W}$ \cite{Coley},
which is a dimensionless measure of the Weyl curvature tensor,
\begin{equation}
\mathcal{W}^{2}=\frac{W^{2}}{H^{4}}=\frac{1}{6H^{4}}\left(  E_{ij}%
E^{ij}+H_{ij}H^{ij}\right)  =\frac{I_{3}}{24H^{4}}, \label{Wp}%
\end{equation}
can be regarded as describing the intrinsic anisotropy in the gravitational
field \cite{W}. Cosmological observations can, in principle, give an upper
bound on $\mathcal{W}$, although obtaining a strong bound is beyond the reach
of present-day observations.

Finally, we shall calculate the gravitational entropy. From a thermodynamic
point of view there is every indication that the entropy of the universe is
\textit{increasing. }Increasing gravitational entropy would naturally be
reflected by increasing local anisotropy, and the Weyl tensor reflects this.
One suggestion in this connection was Penrose's formulation of what is called
the \textit{Weyl curvature conjecture }(WCC) \cite{penrose}. The hypothesis is
motivated by the need for a low entropy constraint on the initial state of the
universe when the matter content was in thermal equilibrium. Penrose has
argued that the low entropy constraint follows from the existence of the
second law of thermodynamics, and that the low entropy in the gravitational
field is tied to constraints on the Weyl curvature. Wainwright and Anderson
\cite{wa} express this conjecture in terms of the ratio of the Weyl and the
Ricci curvature invariants,
\begin{equation}
P^{2}=\frac{I_{3}}{I_{2}}. \label{entropy}%
\end{equation}

The physical content of the conjecture is that the initial state of the
universe is homogeneous and isotropic. As pointed out by Rothman and Anninos
\cite{ra,rothman} (see also \cite{Wain}) the entities $P^{2}$ and $I_{3}$ are
\textquotedblleft\emph{local\textquotedblright} entities in contrast to what
we usually think of entropy. Gr\o n and Hervik (\cite{gron1,gron2}) have
introduced a non-local quantity which shows a more promising behaviour
concerning the WCC. This quantity is also constructed in terms of the Weyl
tensor, and it has therefore a direct connection with the Weyl curvature
tensor but in a \textquotedblleft\emph{non-local form\textquotedblright}.

For SS spacetimes, Pelavas and Lake (\cite{lake}) have pointed out the idea
that Eq. (\ref{entropy}) is not an acceptable candidate for gravitational
entropy along the homothetic trajectories of any self-similar spacetime. Nor
indeed is any \textquotedblleft dimensionless" scalar. It is showed that the
Lie derivative of any "dimensionless" scalar along a homothetic vector field
(HVF) is zero, and concluded that such functions are not acceptable candidates
for the gravitational entropy. Nevertheless \cite{PC}, since self-similar
spacetimes represent asymptotic equilibrium states (since they describe the
asymptotic properties of more general models), and the result $P^{2}=const.,$
is perhaps consistent with this interpretation since the entropy does not
change in these equilibrium models, and perhaps consequently supports the idea
that $P^{2}$ represents a \textquotedblleft gravitational entropy". As we
shall show $\mathcal{W}^{2}$ and $P^{2}$ will be constant along homothetic
trajectories, since all the dimensionless quantities remain constant along
timelike homothetic trajectories.

\subsection{The homothetic vector field}

The homothetic vector field (HVF) is calculated from equation
\begin{equation}
\mathcal{L}_{HO}g_{ij}=2g_{ij}, \label{SS_Eq}%
\end{equation}
(see for example \cite{CC}-\cite{HW} and \cite{Griego}). Algebra brings us to
obtain the following HVF:%
\begin{equation}
HO=\left(  t+t_{0}\right)  \partial_{t}+\left(  1-\left(  t+t_{0}\right)
\frac{a^{\prime}}{a}\right)  x\partial_{x}+\left(  1-\left(  t+t_{0}\right)
\frac{b^{\prime}}{b}\right)  y\partial_{y}+\left(  1-\left(  t+t_{0}\right)
\frac{d^{\prime}}{d}\right)  z\partial_{z}, \label{HO}%
\end{equation}
with the following constrains for the scale factors:
\begin{equation}
a(t)=a_{0}\left(  t+t_{0}\right)  ^{a_{1}},\quad b(t)=b_{0}\left(
t+t_{0}\right)  ^{a_{2}},\quad d(t)=d_{0}\left(  t+t_{0}\right)  ^{a_{3}},
\end{equation}
where $a_{1},a_{2},a_{3}\in\mathbb{R},$ and the following restrictions for the
constants $a_{1},a_{2},a_{3}$ (obtained from the Eq.(\ref{SS_Eq}))%
\begin{equation}
a_{1}=1,\qquad a_{2}=a_{3}. \label{restrictions}%
\end{equation}

As is observed we have been able to obtain a non-singular solution for the
scale factors. Therefore the resulting homothetic vector field is: $HO=\left(
t+t_{0}\right)  \partial_{t}+\left(  1-a_{2}\right)  y\partial_{y}+\left(
1-a_{2}\right)  z\partial_{z}.$

Since we already know how the scale factors behave, then we may calculate all
the curvature invariants as well as all the kinematical quantities.
\begin{equation}
H=\frac{1+2a_{2}}{3\left(  t+t_{0}\right)  },\quad q=2\frac{1-a_{2}}{2a_{2}%
+1},\quad\sigma^{2}=\frac{\sqrt{6}}{3}\frac{\left(  a_{2}-1\right)  ^{2}%
}{\left(  t+t_{0}\right)  ^{2}},
\end{equation}
finding that%
\begin{equation}
\mathcal{A}=\frac{\left(  a_{2}-1\right)  ^{2}}{\left(  1+2a_{2}\right)  ^{2}%
}=const.,
\end{equation}
where, as we can see, $\ \mathcal{A}\in\left(  0,1\right)  ,\forall a_{2}%
\in\left(  0,1\right)  .$ Although the quantity $\mathcal{A}$ is constant,
this quantity may take very small values, in fact $\mathcal{A}$ may runs to
zero. It is also observed that the model never inflates since $q\in\left(
0,2\right)  ,$ $\forall a_{2}\in\left(  0,1\right)  .$

Concerning the curvature invariants we find that%
\begin{align}
I_{0}  &  =\left(  6a_{2}^{2}-2m^{2}\right)  \left(  t+t_{0}\right)
^{-2},\nonumber\\
I_{1}  &  =4\left(  3\left(  a_{2}^{4}+m^{4}\right)  -4a_{2}^{3}+2a_{2}%
^{2}\left(  1-m^{2}\right)  +4m^{2}\left(  a_{2}-1\right)  \right)  \left(
t+t_{0}\right)  ^{-4},\nonumber\\
I_{2}  &  =4\left(  2a_{2}^{2}-2a_{2}^{3}+3a_{2}^{4}-2a_{2}m^{2}+m^{4}\right)
\left(  t+t_{0}\right)  ^{-4},\nonumber\\
I_{3}  &  =16m^{2}\left(  m^{2}+6a_{2}-3\left(  1+a_{2}^{2}\right)  \right)
\left(  t+t_{0}\right)  ^{-4}/3,\nonumber\\
I_{4}  &  =2m^{4}\left(  t+t_{0}\right)  ^{-4}/3,\nonumber\\
I_{5}  &  =2m^{2}\left(  a_{2}-1\right)  ^{2}\left(  t+t_{0}\right)  ^{-4},
\end{align}
and
\begin{align}
\mathcal{W}^{2}  &  =\frac{m^{2}\left(  3\left(  1+a_{2}^{2}\right)
+m^{2}-6a_{2}\right)  }{9\left(  1+2a_{2}\right)  ^{4}}=const,\\
P^{2}  &  =\frac{4m^{2}\left(  m^{2}+6a_{2}-3\left(  1+a_{2}^{2}\right)
\right)  }{3\left(  2a_{2}^{2}-2a_{2}^{3}+3a_{2}^{4}+\left(  m^{2}%
-2a_{2}\right)  m^{2}\right)  }=const.
\end{align}
As above, although $\mathcal{W}^{2}\mathcal{=}const\ll1,$ dimensionless
quantity, it may take very small values, for example, if $a_{2}\rightarrow1$
and $m\rightarrow0$ ($\mathcal{W}^{2}\rightarrow m^{4}/3^{6})$.

\section{The classical model}

We shall take into account the Einstein's field equations (FE) written in the
following form:%
\begin{equation}
R_{ij}-\frac{1}{2}Rg_{ij}=8\pi GT_{ij}^{m}-\Lambda g_{ij}, \label{EFE}%
\end{equation}
where, $\Lambda$ is the cosmological constant and $T_{ij}^{m},$ is the
energy-momentum tensor defined by%
\begin{equation}
T_{ij}^{m}=\left(  p_{m}+\rho_{m}\right)  u_{i}u_{j}+p_{m}g_{ij},
\label{PF-tensor}%
\end{equation}
and where the $4-$velocity is defined by: $u_{i}=\left(  1,0,0,0\right)  ,$
$\rho_{m}$ is the energy density and $p_{m}$ is the pressure. They are related
by the equation: $p_{m}=\omega\rho_{m},$ with $\omega\in(-1,1].$ In this
section we study three models; vacuum solutions, a perfect fluid model and a
model with a perfect fluid and where the constants $G$ and $\Lambda$ are time
varying function.

\subsection{Vacuum solution}

In this case we have found only one solution $a_{2}=m=0.$ Therefore, the
metric Eq. (\ref{BVIo-metric}) collapses to this one:
\begin{equation}
ds^{2}=-\mathrm{dt}^{2}+\left(  t+t_{0}\right)  ^{2}\mathrm{dx}^{2}%
+\mathrm{dy}^{2}+\mathrm{dz}^{2}.
\end{equation}

We may compare this solution with the one obtained in the paper
\cite{Tony1} where we were able to obtain a new solution belonging
to Bianchi $\mathrm{VI}$ type. In this case, this solution does
not belong to Bianchi $\mathrm{VI}_{0}$ type, so we may say that
the metric Eq. (\ref{BVIo-metric}) is more restrictive than the
employed one in \cite{Tony1}. This solution is known as the Taub
form of flat space-time (\cite{WE} chap. 9).

\subsection{Perfect Fluid}

For this model we obtain the following results:%
\begin{equation}
a_{2}=\frac{1-\omega}{2\left(  \omega+1\right)  }\in\left(  0,1\right)  ,\quad
m=\frac{1}{2}\sqrt{\frac{-\left(  3\omega+1\right)  \left(  \omega-1\right)
}{\left(  \omega+1\right)  ^{2}}}\in\left(  0,\frac{1}{2}\right]  ,
\end{equation}
where $\omega\in\left(  -\frac{1}{3},1\right)  ,$ while the scale factors and
the energy density behave as
\begin{equation}
a(t)=a_{0}\left(  t+t_{0}\right)  ,\;b(t)=b_{0}\left(  t+t_{0}\right)
^{a_{2}},\;d=d_{0}\left(  t+t_{0}\right)  ^{a_{2}},\;\rho=\rho_{0}\left(
t+t_{0}\right)  ^{-\gamma},
\end{equation}
where $\gamma=\left(  1+\omega\right)  \left(  1+2a_{2}\right)  ,$ $\rho
_{0}=\frac{A}{8\pi G},$ $A=2a_{2}+a_{2}^{2}-m^{2}.$ With regard to the
deceleration parameter%
\begin{equation}
q=\frac{1}{2}\left(  1+3\omega\right)  >0\qquad\mathcal{A}=\frac{\left(
3\omega+1\right)  ^{2}}{16}=const.\in\left(  0,1\right)  ,
\end{equation}
$\forall\omega\in\left(  -\frac{1}{3},1\right)  ,$ so this model does not
inflate, and $\mathcal{A\rightarrow}0$ only when $\omega\rightarrow-\frac
{1}{3},$ while the Weyl parameter and the gravitational entropy behaves as
\begin{align}
\mathcal{W}^{2}  &  =-\frac{\left(  3\omega+1\right)  ^{2}\left(
\omega-1\right)  \left(  2\omega+1\right)  }{576}=const\in\left(
0,0.012\right)  ,\\
P^{2}  &  =\frac{2}{3}\frac{\left(  3\omega+1\right)  ^{2}\left(
5\omega+1\right)  }{\left(  \omega-1\right)  \left(  3\omega^{2}+1\right)
}=const,\in(-\infty,a],{\quad}a\rightarrow0^{+},
\end{align}
where as it is observed, $\mathcal{W}^{2}\leq0.01,$ it takes a
very small values, $\mathcal{W}^{2}\ll1$, and it runs to zero if
$\omega\rightarrow1$ and $\omega\rightarrow-1/3.$ Notice that our
solution is only valid if $\omega \in\left(  -\frac{1}{3},1\right)
.$ Nevertheless $P^{2}$ has a very pathological behaviour. For
example, $P^{2}\rightarrow a=2.5000\times10^{-11}$ as
$\omega\rightarrow-1/3$, $P^{2}=0$ when $\omega=-\frac{1}{3},$ and
$\omega=-\frac{1}{5}$ but $P^{2}\rightarrow-\infty$ when
$\omega\rightarrow1.$

Therefore, the metric collapses to this one:
\begin{equation}
ds^{2}=-\mathrm{dt}^{2}+\left(  t+t_{0}\right)  ^{2}\mathrm{dx}^{2}+\left(
t+t_{0}\right)  ^{2a_{2}}\left(  \cosh2mx\mathrm{dy}^{2}-2\sinh
2mx\mathrm{dydz}+\cosh2mx\mathrm{dz}^{2}\right)  . \label{SSmetric}%
\end{equation}
Note that in \cite{Tony1} we obtain two solutions, while with the metric Eq.
(\ref{BVIo-metric}) we are only able to obtain one solution which coincides
with the one obtained in \cite{Tony1}. Therefore, this solution is valid when
$\omega\in\left(  -\frac{1}{3},1\right)  $ and $m\in\left(  0,\frac{1}%
{2}\right]  .$ It does not accelerate since $q>0.$ Nevertheless, we may say
that the solution isotropizes since $\mathcal{A\rightarrow}0$ when
$\omega\rightarrow-\frac{1}{3}$ (then $a_{2}\rightarrow1=a_{1}$ and
$m\rightarrow0$) and $\mathcal{W}^{2}\ll1.$ The behaviour of $P^{2}$ shows us,
that maybe, it is not a good definition for the gravitational entropy (at
least in the framework of self-similar solutions) as we have already discussed.

For $\omega=-\frac{1}{3},$ and $\omega=1$ we get that $m=0,$ so the solution
does not belong to Bianchi $\mathrm{VI}_{0}$ type, furthermore if $\omega
=1$,then $a_{2}=0$ i.e. we obtain the vacuum solution.

\subsection{Time varying constants model}

In this framework the FE are the following ones:%
\begin{align}
\frac{a^{\prime}}{a}\frac{b^{\prime}}{b}+\frac{a^{\prime}}{a}\frac{d^{\prime}%
}{d}+\frac{d^{\prime}}{d}\frac{b^{\prime}}{b}-\left(  2+\frac{b^{2}}{d^{2}%
}+\frac{d^{2}}{b^{2}}\right)  \frac{m^{2}}{4a^{2}}  &  =8\pi G\rho_{m}+\Lambda
c^{2},\label{B6NFE1}\\
\frac{b^{\prime\prime}}{b}+\frac{d^{\prime\prime}}{d}+\frac{d^{\prime}}%
{d}\frac{b^{\prime}}{b}+\left(  2+\frac{b^{2}}{d^{2}}+\frac{d^{2}}{b^{2}%
}\right)  \frac{m^{2}}{4a^{2}}  &  =-8\pi G\omega\rho_{m}+\Lambda
c^{2},\label{B6NFE2}\\
\frac{a^{\prime\prime}}{a}+\frac{b^{\prime\prime}}{b}+\frac{a^{\prime}}%
{a}\frac{b^{\prime}}{b}-\left(  2+\frac{3d^{2}}{b^{2}}-\frac{b^{2}}{d^{2}%
}\right)  \frac{m^{2}}{4a^{2}}  &  =-8\pi G\omega\rho_{m}+\Lambda
c^{2},\label{B6NFE3}\\
\frac{b^{\prime\prime}}{b}-\frac{d^{\prime\prime}}{d}+\frac{a^{\prime}}%
{a}\frac{b^{\prime}}{b}-\frac{a^{\prime}}{a}\frac{d^{\prime}}{d}+m^{2}\left(
\frac{b^{2}}{d^{2}a^{2}}-\frac{d^{2}}{b^{2}a^{2}}\right)   &
=0,\label{B6NFE4}\\
\frac{d^{\prime\prime}}{d}+\frac{a^{\prime\prime}}{a}+\frac{a^{\prime}}%
{a}\frac{d^{\prime}}{d}-\left(  2+\frac{3b^{2}}{d^{2}}-\frac{d^{2}}{b^{2}%
}\right)  \frac{m^{2}}{4a^{2}}  &  =-8\pi G\omega\rho_{m}+\Lambda
c^{2},\label{B6NFE5}\\
\rho^{\prime}+\rho\left(  1+\omega\right)  \left(  \frac{a^{\prime}}{a}%
+\frac{b^{\prime}}{b}+\frac{d^{\prime}}{d}\right)   &  =0,\label{B6NFE6}\\
\Lambda^{\prime}  &  =-8\pi G^{\prime}\rho_{m}, \label{B6NFE7}%
\end{align}
where we have taken into account the condition $divT=0.$

Now, we shall take into account the obtained SS restrictions for the scale
factors given by Eq. (\ref{restrictions}). From Eq. (\ref{B6NFE6}) we get%
\begin{equation}
\rho=\rho_{0}\left(  t+t_{0}\right)  ^{-\gamma}, \label{t_den}%
\end{equation}
where $\gamma=\left(  \omega+1\right)  h$ and $h=1+2a_{2}$, since $a_{1}=1$
and $a_{2}=a_{3}.$

From Eq. (\ref{B6NFE1}) we obtain:%
\begin{equation}
\Lambda=\left[  A\left(  t+t_{0}\right)  ^{-2}-8\pi G\rho_{0}\left(
t+t_{0}\right)  ^{-\left(  \omega+1\right)  h}\right]  , \label{peta1}%
\end{equation}
where $A=2a_{2}+a_{2}^{2}-m^{2}.$ Now, taking into account Eq. (\ref{B6NFE7})
and Eq. (\ref{peta1}), algebra brings us to obtain%
\begin{equation}
G=G_{0}\left(  t+t_{0}\right)  ^{\gamma-2},\qquad G_{0}=\frac{A}{4\pi\rho
_{0}\left(  \omega+1\right)  }, \label{G}%
\end{equation}
while the cosmological \textquotedblleft constant\textquotedblright\ behaves
as:
\begin{equation}
\Lambda=\frac{A}{c^{2}}\left(  1-\frac{2}{\gamma}\right)  \left(
t+t_{0}\right)  ^{-2}=\Lambda_{0}\left(  t+t_{0}\right)  ^{-2}.
\end{equation}

In this case we have found the following solution%
\begin{equation}
a_{2\pm}=\frac{1}{2}\left(  1\pm\sqrt{1-4m^{2}}\right)  ,\qquad\forall
m\in\left[  -\frac{1}{2},\frac{1}{2}\right]  \backslash\left\{  0\right\}  ,
\end{equation}
$a_{2+}\in\lbrack1/2,1),$ $a_{2-}\in\lbrack0,1/2),$ hence $h_{\pm}=1+2a_{2\pm
}$ and therefore we obtain $h_{+}\in\lbrack2,3)$ and $h_{-}\in\lbrack1,2)$%
\begin{align}
a(t)  &  =a_{0}\left(  t+t_{0}\right)  ,\quad b(t)=b_{0}\left(  t+t_{0}%
\right)  ^{a_{2\pm}},\quad d=d_{0}\left(  t+t_{0}\right)  ^{a_{2\pm}%
},\nonumber\\
\rho &  =\rho_{0}\left(  t+t_{0}\right)  ^{-\gamma_{\pm}},\;G=G_{0\pm}\left(
t+t_{0}\right)  ^{\gamma_{\pm}-2},\;\Lambda=\Lambda_{0\pm}\left(
t+t_{0}\right)  ^{-2},
\end{align}
with $\gamma_{\pm}=\left(  \omega+1\right)  h_{\pm}.$ Notice that this
solution is valid $\forall\omega\in(-1,1].$ In this way the metric collapses
to Eq. (\ref{SSmetric}). The behaviour of the \textquotedblleft
constants\textquotedblright\ is the following one:%
\begin{equation}
G\thickapprox\left\{
\begin{array}
[c]{lll}%
decreasing & if & \left(  \omega+1\right)  h_{\pm}<2\\
constant & if & \left(  \omega+1\right)  h_{\pm}=2\\
growing & if & \left(  \omega+1\right)  h_{\pm}>2
\end{array}
\right.  ,\qquad\Lambda_{0\pm}\thickapprox\left\{
\begin{array}
[c]{ccc}%
<0 & if & \left(  \omega+1\right)  h_{\pm}<2\\
=0 & if & \left(  \omega+1\right)  h_{\pm}=2\\
>0 & if & \left(  \omega+1\right)  h_{\pm}>2
\end{array}
\right.  ,
\end{equation}
where $h_{\pm}=1+2a_{2\pm}=2\pm\sqrt{1-4m^{2}}.$ Note that in \cite{Tony1} we
obtained another solution. With regard to the deceleration parameter (for
simplicity we have performed all these calculations with $h_{+}$ i.e. with
$a_{2+}$)
\begin{align}
q_{+}  &  =\frac{3}{h_{+}}-1>0,\qquad q\in\left(  0,\frac{1}{2}\right)
,\quad\forall m\in\left[  -\frac{1}{2},\frac{1}{2}\right]  \backslash\left\{
0\right\}  ,\\
\mathcal{A}  &  =\frac{1}{4}\frac{\left(  \sqrt{1-4m^{2}}-1\right)  ^{2}%
}{\left(  \sqrt{1-4m^{2}}+2\right)  ^{2}}=const\in\left(  0,0.06\right)
\rightarrow0,
\end{align}%
\begin{align}
\mathcal{W}^{2}  &  =-\frac{m^{2}\left(  4m^{2}-3+3\sqrt{1-4m^{2}}\right)
}{18\left(  2+\sqrt{1-4m^{2}}\right)  }=const\in\left(  0,0.0016\right)  ,\\
P^{2}  &  =-\frac{4m^{2}\left(  8m^{2}-3+3\sqrt{1-4m^{2}}\right)  }{3\left(
\left(  6m^{2}-3\right)  \sqrt{1-4m^{2}}+12m^{2}-8m^{4}-3\right)  }=const,
\end{align}
where $P^{2}\in\left(  -\infty,0.01\right)  ,$ $\forall m\in\left[  -\frac
{1}{2},\frac{1}{2}\right]  \backslash\left\{  0\right\}  .$ Therefore this
solution is valid $\forall\omega\in(-1,1]$ and $\forall m\in\left(  -\frac
{1}{2},\frac{1}{2}\right)  \backslash\left\{  0\right\}  .$ The model does not
accelerate but isotropizes since $\mathcal{W}^{2}\rightarrow0$ as well as
$\mathcal{A\ll}1.$ With regard to the quantity $P^{2}$ it is observed that
$P^{2}\rightarrow0$, $\forall m\in\left(  -\frac{1}{2},\frac{1}{2}\right)
\backslash\left\{  0\right\}  $ and only runs to minus infinity when
$m\rightarrow\pm\frac{1}{2}.$

\section{Scalar field model}

In this section we are going to study several scalar models. In the first
place we study which kinds of potentials are compatible with the self-similar
solution. For this purpose we study through the Lie group method the resulting
Klein-Gordon equation. Once we have deduced the potential compatible with the
self-similar solution we study if this kind of potential brings us to obtain
self-similar solution. We answer in this case is no, we only obtain power law
solutions but this fact does not mean that they are self-similar solution.
Therefore, after this brief analysis on the potential, we start by studying a
simple scalar model. In second place we study a non-interacting scalar and
matter model. We leave for a forthcoming paper the study of the very
interesting case of interacting scalar and matter models (see for example
\cite{Wetterich} and \cite{Coleys}). In the third class of studied models we
introduce the hypothesis of a $G-$var, i.e. we study a scalar model where
$G=G(t),$ is a function of time. In this case, in a phenomenological way, we
outline a modified Klein-Gordon equation in order to take into account the
possible variation on time of the function $G(t).$ We go next to study the
kind of potential compatible with a self-similar solution and $G(t)$. To end
this section, we study a model with scalar and matter fields and $G-$varying.

The stress-energy tensor may be written in the following form:%
\begin{equation}
T_{ij}^{\phi}=\left(  p^{\phi}+\rho^{\phi}\right)  u_{i}u_{j}+p^{\phi}g_{ij},
\label{STensor}%
\end{equation}
where the energy density and the pressure of the fluid due a scalar field are
given by%
\begin{equation}
\rho^{\phi}=\frac{1}{2}\phi^{\prime2}+V(\phi),\qquad p^{\phi}=\frac{1}{2}%
\phi^{\prime2}-V(\phi),
\end{equation}
while the conservation equation reads (Klein-Gordon equation)%
\begin{equation}
\phi^{\prime\prime}+\phi^{\prime}H+\frac{d}{d\phi}V=0. \label{KG1}%
\end{equation}

We need to study the class of potential compatible with the SS solution, for
this reason we study by using the Lie group method the KG equation, (for an
introduction to the Lie group method see for example \cite{Ibra}-\cite{bluman}
and \cite{Tony2} for a concrete application in cosmological contexts). In
particular, we seek the forms of $V(\phi)$ for which our field equations admit
symmetries and therefore they are integrable. In this case we already know
that the Hubble function behaves as: $H=h\left(  t+t_{0}\right)  ^{-1},$
$h\in\mathbb{R}^{+}$, so the KG equation reads%
\begin{equation}
\phi^{\prime\prime}+h\phi^{\prime}\left(  t+t_{0}\right)  ^{-1}+\frac
{dV}{d\phi}=0.
\end{equation}

\begin{theorem}
The only possible form for the potential $V\left(  \phi\right)  $ for a
spacetime admitting a HFV, $HO,$ is $V(\phi)=V_{0}\exp\left(  \kappa
\phi\right)  $ and therefore $\phi=\alpha\ln t.$
\end{theorem}

\begin{proof}
The application of the Lie group method brings us to outline the following
system of PDEs%
\begin{align}
\xi_{\phi\phi}t^{2} =0,\label{iga1}\\
2ht\xi_{\phi}+t^{2}\eta_{\phi\phi}-2t^{2}\xi_{t\phi} =0,\label{iga2}\\
3t^{2}\xi_{\phi}\frac{dV}{d\phi}+ht\xi_{t}-h\xi+2t^{2}\eta_{t\phi}-t^{2}%
\xi_{tt}=0,\label{iga3}\\
t^{2}\eta\frac{d^{2}V}{d\phi^{2}}+t^{2}\eta_{tt}+2t^{2}\xi_{t}\frac{dV}{d\phi
}-t^{2}\eta_{\phi}\frac{dV}{d\phi}+3at\eta_{t}=0, \label{iga4}%
\end{align}
If we impose the symmetry $\xi=\alpha\left(  t+t_{0}\right)  ,$ $\eta=\delta,$
then its invariant solution is: $\phi=\frac{\delta}{\alpha}\ln\frac{1}{\alpha
}\left(  t+t_{0}\right)  ,$ then, from Eq. (\ref{iga4}), we obtain the next
restriction for the potential $V$%
\begin{equation}
\delta\frac{d^{2}V}{d\phi^{2}}+2\alpha\frac{dV}{d\phi}=0\;\Longrightarrow\;
V=\beta\exp\left(  -2\frac{\alpha}{\delta}\phi\right)  +\kappa,\;\alpha
,\beta,\delta,\kappa\in\mathbb{R}.
\end{equation}

Therefore we have found, redefining the numerical constants, that the only
solution compatible with the FE is%
\begin{equation}
\phi=\pm\sqrt{\alpha}\ln\left(  t+t_{0}\right)  ,\qquad V=\beta\exp\left(
\mp\frac{2}{\sqrt{\alpha}}\phi\right)  . \label{potential1}%
\end{equation}
as it is required.
\end{proof}

Note that in this case it is possible to find more symmetries, but the
solution generated by them are not compatible with the FE. For example, if we
impose the symmetry, $\xi=\alpha t,$ $\eta=\delta\phi$, then Eq. (\ref{iga4})
yields%
\begin{equation}
\delta\phi\frac{d^{2}V}{d\phi^{2}}+\left(  2\alpha-\delta\right)  \frac
{dV}{d\phi}=0\quad\Longrightarrow\quad V=\kappa_{1}\phi^{-\frac{2}{\delta
}\left(  \alpha-\delta\right)  }+\kappa_{2}, \label{potential2}%
\end{equation}
which is the potential proposed by Peebles and Ratra, $V\thickapprox
\phi^{-\alpha}$ \cite{quintessence}, but this solution it is not compatible
with the FE with a SS solution. Nevertheless we shall use this potential in
the $G-$varying scenario. In the appendix we give an alternative derivation of
all these results by using the matter collineation approach following a
previous paper (see \cite{Tony2}).

Models with a self-interaction potential with an exponential dependence on the
scalar field of the form $V=\beta\exp\left(  \mp2\phi\right)  ,$ have been the
subject of much interest and arise naturally from theories of gravity such as
scalar-tensor theories or string theories \cite{Green88}. Recently, it has
been argued that a scalar field with an exponential potential is a strong
candidate for dark matter in spiral galaxies \cite{Guzman} and is consistent
with observations of current accelerated expansion of the universe
\cite{Huterer}.

In the inverse way we may state the following theorem.

\begin{theorem}
For a scalar model if the potential is of the form $V=\beta\exp\left(
\mp2\phi\right)  ,$ then the scale factors must follow a power law solution
i.e. $H=ht^{-1}.$
\end{theorem}

\begin{proof}
As above we perform the proof by using the Lie group method. In this case we
must study the following ODE%
\[
\phi^{\prime\prime}+\phi^{\prime}H+\frac{d}{d\phi}V=0,
\]
where $V=\beta\exp\left(  \mp2\phi\right)  ,$ then we shall study the
different forms for the function $H(t)$ in order to get and integrable ODE.

We have the next system of PDEs%
\begin{align}
\xi_{\phi\phi}  &  =0,\\
2H\xi_{\phi}+\eta_{\phi\phi}-2\xi_{t\phi}  &  =0,\\
-6e^{-2\phi}\xi_{\phi}+H\xi_{t}-H^{\prime}\xi+2\eta_{t\phi}-\xi_{tt}  &
=0,\label{iveta}\\
4e^{-2\phi}\eta+\eta_{tt}-4e^{-2\phi}\xi_{t}+2e^{-2\phi}\eta_{\phi}+H\eta_{t}
&  =0,
\end{align}
As we can easily see, the symmetry $\xi=t,\eta=1,$ brings us to get $\phi=\ln
t,$ as invariant solution, and from Eq. (\ref{iveta}) we obtain the result
i.e. $H=ht^{-1},$ $h\in\mathbb{R}.$
\end{proof}

Obviously this result does not mean that the solution must be self-similar, it
only means that the scale factor must follow a power-law solution i.e. they
are of the form $a_{i}(t)=a_{0}t^{a_{j}},$ with $a_{j}\in\mathbb{R}^{+}.$

\subsection{Scalar model}

We write the FE in the following form:%
\begin{align}
\frac{a^{\prime}}{a}\frac{b^{\prime}}{b}+\frac{a^{\prime}}{a}\frac{d^{\prime}%
}{d}+\frac{d^{\prime}}{d}\frac{b^{\prime}}{b}-\left(  2+\frac{b^{2}}{d^{2}%
}+\frac{d^{2}}{b^{2}}\right)  \frac{m^{2}}{4a^{2}}  &  =8\pi G\rho_{\phi},\\
\frac{b^{\prime\prime}}{b}+\frac{d^{\prime\prime}}{d}+\frac{d^{\prime}}%
{d}\frac{b^{\prime}}{b}+\left(  2+\frac{b^{2}}{d^{2}}+\frac{d^{2}}{b^{2}%
}\right)  \frac{m^{2}}{4a^{2}}  &  =-8\pi Gp_{\phi},\\
\frac{a^{\prime\prime}}{a}+\frac{b^{\prime\prime}}{b}+\frac{a^{\prime}}%
{a}\frac{b^{\prime}}{b}-\left(  2+\frac{3d^{2}}{b^{2}}-\frac{b^{2}}{d^{2}%
}\right)  \frac{m^{2}}{4a^{2}}  &  =-8\pi Gp_{\phi},\\
\frac{b^{\prime\prime}}{b}-\frac{d^{\prime\prime}}{d}+\frac{a^{\prime}}%
{a}\frac{b^{\prime}}{b}-\frac{a^{\prime}}{a}\frac{d^{\prime}}{d}+m^{2}\left(
\frac{b^{2}}{d^{2}a^{2}}-\frac{d^{2}}{b^{2}a^{2}}\right)   &  =0,\\
\frac{d^{\prime\prime}}{d}+\frac{a^{\prime\prime}}{a}+\frac{a^{\prime}}%
{a}\frac{d^{\prime}}{d}-\left(  2+\frac{3b^{2}}{d^{2}}-\frac{d^{2}}{b^{2}%
}\right)  \frac{m^{2}}{4a^{2}}  &  =-8\pi Gp_{\phi},\\
\phi^{\prime\prime}+\phi^{\prime2}H+\frac{d}{d\phi}V  &  =0.
\end{align}
By assuming the potential given by Eq. (\ref{potential1}) it is possible to
find the next set of solutions
\begin{align}
a_{2\pm}  &  =\frac{1}{2}\left(  1\pm\sqrt{1-4m^{2}}\right)  ,\qquad\forall
m\in\left[  -\frac{1}{2},\frac{1}{2}\right]  \backslash\left\{  0\right\}
,\nonumber\\
a(t)  &  =a_{0}\left(  t+t_{0}\right)  ,\quad b(t)=b_{0}\left(  t+t_{0}%
\right)  ^{a_{2\pm}},\quad d=d_{0}\left(  t+t_{0}\right)  ^{a_{2\pm}},
\end{align}
and%
\begin{align}
\alpha_{\pm}  &  =1\pm\sqrt{1-4m^{2}},\quad\beta_{\pm}=\frac{1}{2}\left(
1\pm\sqrt{1-4m^{2}}\right)  ^{2},\\
\phi &  =\pm\sqrt{\alpha_{\pm}}\ln\left(  t+t_{0}\right)  ,\qquad V=\beta
_{\pm}\exp\left(  \mp\frac{2}{\sqrt{\alpha_{\pm}}}\phi\right)  .
\end{align}

As it is observed, we have obtained the same behavior for the scale factor as
the one obtained in the case of a perfect fluid with time-varying constants
model. For this reason, as we already know, we get: $q>0,$ $\forall
m\in\left[  -\frac{1}{2},\frac{1}{2}\right]  \backslash\left\{  0\right\}  ,$
$\mathcal{A}=const\in\left(  0,0.06\right)  \rightarrow0,$ while the Weyl
parameter and the gravitational entropy behaves as $\mathcal{W}^{2}%
=const\in\left(  0,0.0016\right)  \ll1,$ and $P^{2}=const\in\left(
-\infty,0.01\right)  ,$ $\forall m\in\left[  -\frac{1}{2},\frac{1}{2}\right]
\backslash\left\{  0\right\}  .$ Therefore the model does not accelerate but
isotropizes since the quantities ($\mathcal{A}$ and $\mathcal{W}^{2}$) instead
of being constant, they take values very close to zero. With regard to the
quantity $P^{2}$ it is observed that $P^{2}\rightarrow0$, (it takes values
very close to zero) $\forall m\in\left(  -\frac{1}{2},\frac{1}{2}\right)
\backslash\left\{  0\right\}  $ and it only runs to minus infinity when
$m\rightarrow\pm\frac{1}{2}.$

\subsection{Non-interacting scalar and matter fields}

The stress-energy tensor may be written in the following form: $T=T^{m}%
+T^{\phi},$ where the energy density and the pressure of the fluid due a
scalar field are given by Eq. (\ref{STensor}). This describe a non-interacting
dark matter and dark energy cosmological model (we assume that the baryon
component can be ignored). Since the nature of both dark energy and dark
matter is still unknown, there is no physical argument to exclude the possible
non-interaction between them.

We write the FE in the following form:%
\begin{align}
\frac{a^{\prime}}{a}\frac{b^{\prime}}{b}+\frac{a^{\prime}}{a}\frac{d^{\prime}%
}{d}+\frac{d^{\prime}}{d}\frac{b^{\prime}}{b}-\left(  2+\frac{b^{2}}{d^{2}%
}+\frac{d^{2}}{b^{2}}\right)  \frac{m^{2}}{4a^{2}}  &  =8\pi G\left(  \rho
_{m}+\rho_{\phi}\right)  ,\\
\frac{b^{\prime\prime}}{b}+\frac{d^{\prime\prime}}{d}+\frac{d^{\prime}}%
{d}\frac{b^{\prime}}{b}+\left(  2+\frac{b^{2}}{d^{2}}+\frac{d^{2}}{b^{2}%
}\right)  \frac{m^{2}}{4a^{2}}  &  =-8\pi G\left(  \omega\rho_{m}+p_{\phi
}\right)  ,\\
\frac{a^{\prime\prime}}{a}+\frac{b^{\prime\prime}}{b}+\frac{a^{\prime}}%
{a}\frac{b^{\prime}}{b}-\left(  2+\frac{3d^{2}}{b^{2}}-\frac{b^{2}}{d^{2}%
}\right)  \frac{m^{2}}{4a^{2}}  &  =-8\pi G\left(  \omega\rho_{m}+p_{\phi
}\right)  ,\\
\frac{b^{\prime\prime}}{b}-\frac{d^{\prime\prime}}{d}+\frac{a^{\prime}}%
{a}\frac{b^{\prime}}{b}-\frac{a^{\prime}}{a}\frac{d^{\prime}}{d}+m^{2}\left(
\frac{b^{2}}{d^{2}a^{2}}-\frac{d^{2}}{b^{2}a^{2}}\right)   &  =0,\\
\frac{d^{\prime\prime}}{d}+\frac{a^{\prime\prime}}{a}+\frac{a^{\prime}}%
{a}\frac{d^{\prime}}{d}-\left(  2+\frac{3b^{2}}{d^{2}}-\frac{d^{2}}{b^{2}%
}\right)  \frac{m^{2}}{4a^{2}}  &  =-8\pi G\left(  \omega\rho_{m}+p_{\phi
}\right)  ,
\end{align}
and the conservation equations now read
\begin{equation}
\rho_{m}^{\prime}+\left(  \omega+1\right)  \rho_{m}H=0, \label{c1}%
\end{equation}
and%
\begin{equation}
\phi^{\prime\prime}+\phi^{\prime}H+\frac{d}{d\phi}V=0. \label{c2}%
\end{equation}
where $H=h\left(  t+t_{0}\right)  ^{-1},$ and $h=1+2a_{2}.$

In this case we have found the next solutions
\begin{align}
a_{2}  &  =\frac{1-\omega}{2\omega+2}\in\left(  0,1\right)  \qquad m=\frac
{1}{2}\frac{\sqrt{-\left(  3\omega+1\right)  \left(  \omega-1\right)  }%
}{1+\omega}\in\left(  0,\frac{1}{2}\right]  ,\nonumber\\
\beta &  =\alpha a_{2}=\alpha\left(  \frac{1-\omega}{2\omega+2}\right)
\qquad\rho_{0}=\frac{1-\omega-\alpha\left(  1+\omega\right)  }{\left(
1+\omega\right)  ^{2}}>0,
\end{align}
where the constant $\alpha$ must verify the condition: $\alpha>\frac{1-\omega
}{\omega+1}>0,$ $\forall\omega\in\left(  -\frac{1}{3},1\right)  .$ Therefore
we have the following behaviour for the main quantities%
\begin{equation}
\rho_{m}=\rho_{0}\left(  t+t_{0}\right)  ^{-\left(  1+\omega\right)  h},\qquad
p_{m}=\omega\rho_{m},
\end{equation}%
\begin{equation}
a_{1}=a_{0}\left(  t+t_{0}\right)  ,\qquad b=b_{0}\left(  t+t_{0}\right)
^{a_{2}},\qquad d=d_{0}\left(  t+t_{0}\right)  ^{a_{2}},
\end{equation}%
\begin{equation}
\phi=\pm\sqrt{\alpha}\ln\left(  t+t_{0}\right)  ,\qquad V=\beta\exp\left(
\mp\frac{2}{\sqrt{\alpha}}\phi\right)  .
\end{equation}

Since the scale factor behaves as in the perfect fluid solution (see above)
then the deceleration parameter (as we already know) behaves as: $q=\frac
{1}{2}\left(  1+3\omega\right)  >0,$ $\mathcal{A}=\frac{\left(  3\omega
+1\right)  ^{2}}{16}=const.\in\left(  0,1\right)  ,$ $\forall\omega\in\left(
-\frac{1}{3},1\right)  .$ Therefore, with the above restrictions on the
$\omega-$parameter our model does not inflate, $q>0.$ While the Weyl parameter
and the gravitational entropy behave as: $\mathcal{W}^{2}=const\in\left(
0,0.012\right)  ,$ and $P^{2}=const,\in(-\infty,a],$ with $a\rightarrow0^{+}$
(see the above discussion about these quantities).

\subsection{$G-$varying}

We would like to study how the gravitational varies constant when we are
considering only a scalar field. For this purpose, in analogy with the perfect
fluid case (see \cite{Tony3}) and in a phenomenological way, by using the
Bianchi identity $div\left(  8\pi G(t)T_{ij}\right)  =0,$ we propose the
following conservation equation%
\begin{equation}
G\rho^{\prime}+G\left(  \rho+p\right)  H=-G^{\prime}\rho\quad
\Longleftrightarrow\quad\phi^{\prime}\left(  \square\phi+\frac{dV}{d\phi
}\right)  =-\frac{G^{\prime}}{G}\rho_{\phi}, \label{Gv1}%
\end{equation}
which is the modified KG equation.

For this model the FE read%
\begin{align}
\frac{a^{\prime}}{a}\frac{b^{\prime}}{b}+\frac{a^{\prime}}{a}\frac{d^{\prime}%
}{d}+\frac{d^{\prime}}{d}\frac{b^{\prime}}{b}-\left(  2+\frac{b^{2}}{d^{2}%
}+\frac{d^{2}}{b^{2}}\right)  \frac{m^{2}}{4a^{2}}  &  =8\pi G(t)\rho_{\phi
},\\
\frac{b^{\prime\prime}}{b}+\frac{d^{\prime\prime}}{d}+\frac{d^{\prime}}%
{d}\frac{b^{\prime}}{b}+\left(  2+\frac{b^{2}}{d^{2}}+\frac{d^{2}}{b^{2}%
}\right)  \frac{m^{2}}{4a^{2}}  &  =-8\pi G(t)p_{\phi},\\
\frac{a^{\prime\prime}}{a}+\frac{b^{\prime\prime}}{b}+\frac{a^{\prime}}%
{a}\frac{b^{\prime}}{b}-\left(  2+\frac{3d^{2}}{b^{2}}-\frac{b^{2}}{d^{2}%
}\right)  \frac{m^{2}}{4a^{2}}  &  =-8\pi G(t)p_{\phi},\\
\frac{b^{\prime\prime}}{b}-\frac{d^{\prime\prime}}{d}+\frac{a^{\prime}}%
{a}\frac{b^{\prime}}{b}-\frac{a^{\prime}}{a}\frac{d^{\prime}}{d}+m^{2}\left(
\frac{b^{2}}{d^{2}a^{2}}-\frac{d^{2}}{b^{2}a^{2}}\right)   &  =0,\\
\frac{d^{\prime\prime}}{d}+\frac{a^{\prime\prime}}{a}+\frac{a^{\prime}}%
{a}\frac{d^{\prime}}{d}-\left(  2+\frac{3b^{2}}{d^{2}}-\frac{d^{2}}{b^{2}%
}\right)  \frac{m^{2}}{4a^{2}}  &  =-8\pi G(t)p_{\phi},\\
\phi^{\prime\prime}+H\phi^{\prime}+\frac{dV}{d\phi}  &  =-\frac{G^{\prime}}%
{G}\frac{1}{\phi^{\prime}}\rho_{\phi}.
\end{align}

In order to solve the FE we need to solve Eq. (\ref{Gv1}). To that end, we
shall study it through the LG method. Eq. (\ref{Gv1}) could be rewritten in
the following form%
\begin{equation}
\phi^{\prime\prime}\phi^{\prime}+ht^{-1}\phi^{\prime2}+\frac{dV}{d\phi}%
\phi^{\prime}+\rho_{\phi}\frac{G^{\prime}}{G}=0, \label{lisa}%
\end{equation}
where $H=ht^{-1}.$ For simplicity, and without lost of generality, we consider
$H=ht^{-1}$ instead of its non-singular form. As above, we are seeking the
forms of $V\left(  \phi\right)  $ and $G(t)$ for which our field equations
admit symmetries and therefore they are integrable

As above, in order to study the possible solutions to Eq. (\ref{lisa}) we
apply the LG method, where the standard procedure brings us to get the
following system of PDEs%
\begin{align}
\xi_{\phi\phi}  &  =0,\label{zoe1}\\
\eta_{\phi\phi}-2\xi_{t\phi}+\left(  2ht^{-1}+\frac{G^{\prime}}{G}\right)
\xi_{\phi}  &  =0,\\
2\eta_{t\phi}+\left(  \frac{1}{2}\left(  \frac{G^{\prime\prime}}{G}-\left(
\frac{G^{\prime}}{G}\right)  ^{2}\right)  -ht^{-2}\right)  \xi+3V_{\phi}%
\xi_{\phi}+\left(  ht^{-1}+\frac{G^{\prime}}{2G}\right)  \xi_{t}-\xi_{tt}  &
=0,\label{zoe3}\\
\eta_{tt}+V_{\phi\phi}\eta+4\frac{G^{\prime}}{G}V\xi_{\phi}+\left(
ht^{-1}+\frac{G^{\prime}}{2G}\right)  \eta_{t}+2V_{\phi}\xi_{t}-V_{\phi}%
\eta_{\phi}  &  =0,\label{zoe4}\\
\left(  \frac{G^{\prime\prime}}{G}-\left(  \frac{G^{\prime}}{G}\right)
^{2}\right)  V\xi+\frac{G^{\prime}}{G}V_{\phi}\eta+3\frac{G^{\prime}}{G}%
V\xi_{t}-2\frac{G^{\prime}}{G}V\eta_{\phi}  &  =0,\label{zoe5}\\
\frac{G^{\prime}}{G}V\eta_{t}  &  =0, \label{zoe6}%
\end{align}
where, $V_{\phi}=\frac{dV}{d\phi}.$ The following symmetry%
\begin{equation}
\xi=\frac{-t}{\alpha},\qquad\eta=\phi\qquad\Longrightarrow\qquad
\phi=t^{-\alpha}%
\end{equation}
then we obtain the following restrictions, from Eqs. (\ref{zoe3}-\ref{zoe5}).
From Eq. (\ref{zoe3}) we get%
\begin{equation}
G^{\prime\prime}=\frac{G^{\prime^{2}}}{G}-\frac{G^{\prime}}{t}\qquad
\Longrightarrow\qquad G=\kappa_{1}t^{g},
\end{equation}
while from Eq. (\ref{zoe4}) we obtain
\begin{equation}
V_{\phi\phi}\phi-\left(  \frac{2}{\alpha}+1\right)  V_{\phi}=0\qquad
\Longrightarrow\qquad V=\kappa_{2}\phi^{2\left(  \frac{1}{\alpha}+1\right)  },
\end{equation}
in this way $V=\kappa_{2}\left(  t+t_{0}\right)  ^{-2\left(  \alpha+1\right)
}.$ So we have found that the main quantities behave as follows
\begin{equation}
\phi=\left(  t+t_{0}\right)  ^{-\alpha},\;V=\kappa_{2}\phi^{2\left(  \frac
{1}{\alpha}+1\right)  }=\kappa_{2}\left(  t+t_{0}\right)  ^{-2\left(
\alpha+1\right)  },\;G=\kappa_{1}\left(  t+t_{0}\right)  ^{g}, \label{Th2}%
\end{equation}
such that $g-2\left(  \alpha+1\right)  =-2,$ and therefore $g=2\alpha.$ Note
that we may redefine the constants in order to get $V\thickapprox\phi
^{-\alpha}.$

\begin{theorem}
The only compatible form for the potential $V\left(  \phi\right)  $ with the
FE for a spacetime admitting a HFV, $HO,$ where $G=G(t)$, is $V(\phi
)=V_{0}\phi^{-\alpha}$ and therefore $\phi=\left(  t+t_{0}\right)  ^{\beta},$
and $G=\kappa_{1}\left(  t+t_{0}\right)  ^{g},$ with $\alpha,\beta
,g\in\mathbb{R}.$
\end{theorem}

Taking into account all these results we have found the next solution%
\begin{align}
a_{2}  &  =\frac{1}{2}\left(  1\pm\sqrt{1-4m^{2}}\right)  ,\qquad\forall
m\in\left[  -\frac{1}{2},\frac{1}{2}\right]  \backslash\left\{  0\right\}  ,\\
\alpha &  =\alpha>0,\;G_{0}=const.>0,\;\beta=\frac{\left(  1-2m^{2}\right)
}{G_{0}}+\sqrt{\left(  1-4m^{2}\right)  },\nonumber
\end{align}
so the behaviour of the main quantities is the following one%
\begin{equation}
a_{1}=a_{0}\left(  t+t_{0}\right)  ,\quad b=b_{0}\left(  t+t_{0}\right)
^{a_{2}},\quad d=d_{0}\left(  t+t_{0}\right)  ^{a_{2}},
\end{equation}%
\begin{equation}
\phi=\phi_{0}\left(  t+t_{0}\right)  ^{-\alpha},\qquad V=\beta\left(
t+t_{0}\right)  ^{-2\left(  \alpha+1\right)  },\qquad G=G_{0}\left(
t+t_{0}\right)  ^{2\alpha}.
\end{equation}

Notice that this is the same solution for the scale factor than in the above
models. Therefore, with the above restrictions on the $m-$parameter our model
does not inflate $q>0.$ i.e. the model does not accelerate but isotropizes
since $\mathcal{W}^{2}\rightarrow0$ as well as $\mathcal{A}.$ With regard to
the quantity $P^{2}$ it is observed that $P^{2}\rightarrow0$, $\forall
m\in\left(  -\frac{1}{2},\frac{1}{2}\right)  \backslash\left\{  0\right\}  $
and it only runs to minus infinity when $m\rightarrow\pm\frac{1}{2}.$ With
regard to the gravitational constant $G,$ we have obtained that it is an
increasing time function since $\alpha>0$.

\subsection{$G$ variable with matter and a scalar field}

We start by rewriting the stress-energy tensor in the following way,
$T^{ij}=T_{m}^{ij}+T_{\phi}^{ij}$, where $T^{ij}=(\tilde{p}+\tilde{\rho}%
)u^{i}u^{j}+\tilde{p}g^{ij},$ and $\tilde{\rho}=\rho_{m}+\rho_{\phi}$ and
$\tilde{p}=p_{m}+p_{\phi},$ and taking into account the Bianchi identity
\begin{equation}
div\left(  8\pi G(t) T_{ij}\right)  =0\quad\Longleftrightarrow\quad
G\tilde{\rho}^{\prime}+G\left(  \tilde{p}+\tilde{\rho}\right)  H=-G^{\prime
}\tilde{\rho},
\end{equation}
i.e.%
\begin{equation}
\rho_{m}^{\prime}+\left(  \rho_{m}+p_{m}\right)  H+\phi^{\prime}\left(
\square\phi+\frac{dV}{d\phi}\right)  =-\frac{G^{\prime}}{G}\left(  \rho
_{m}+p_{\phi}\right)  . \label{mar}%
\end{equation}

We may study Eq. (\ref{mar}) in several ways. One of them, maybe the simplest
one, may be spliting it into%
\begin{align}
\rho_{m}^{\prime}+\left(  \rho_{m}+p_{m}\right)  H  &  =-\frac{G^{\prime}}%
{G}\rho_{m},\\
\phi^{\prime}\left(  \square\phi+\frac{dV}{d\phi}\right)   &  =-\frac
{G^{\prime}}{G}\rho_{\phi}.
\end{align}
Notice that this approach is similar to a scenario describing an interacting
scalar and matter fields.

There the FE for this model are:%
\begin{align}
\frac{a^{\prime}}{a}\frac{b^{\prime}}{b}+\frac{a^{\prime}}{a}\frac{d^{\prime}%
}{d}+\frac{d^{\prime}}{d}\frac{b^{\prime}}{b}-\left(  2+\frac{b^{2}}{d^{2}%
}+\frac{d^{2}}{b^{2}}\right)  \frac{m^{2}}{4a^{2}}  &  =8\pi G(t)\left(
\rho_{m}+\rho_{\phi}\right)  ,\\
\frac{b^{\prime\prime}}{b}+\frac{d^{\prime\prime}}{d}+\frac{d^{\prime}}%
{d}\frac{b^{\prime}}{b}+\left(  2+\frac{b^{2}}{d^{2}}+\frac{d^{2}}{b^{2}%
}\right)  \frac{m^{2}}{4a^{2}}  &  =-8\pi G(t)\left(  \omega\rho_{m}+p_{\phi
}\right)  ,\\
\frac{a^{\prime\prime}}{a}+\frac{b^{\prime\prime}}{b}+\frac{a^{\prime}}%
{a}\frac{b^{\prime}}{b}-\left(  2+\frac{3d^{2}}{b^{2}}-\frac{b^{2}}{d^{2}%
}\right)  \frac{m^{2}}{4a^{2}}  &  =-8\pi G(t)\left(  \omega\rho_{m}+p_{\phi
}\right)  ,\\
\frac{b^{\prime\prime}}{b}-\frac{d^{\prime\prime}}{d}+\frac{a^{\prime}}%
{a}\frac{b^{\prime}}{b}-\frac{a^{\prime}}{a}\frac{d^{\prime}}{d}+m^{2}\left(
\frac{b^{2}}{d^{2}a^{2}}-\frac{d^{2}}{b^{2}a^{2}}\right)   &  =0,\\
\frac{d^{\prime\prime}}{d}+\frac{a^{\prime\prime}}{a}+\frac{a^{\prime}}%
{a}\frac{d^{\prime}}{d}-\left(  2+\frac{3b^{2}}{d^{2}}-\frac{d^{2}}{b^{2}%
}\right)  \frac{m^{2}}{4a^{2}}  &  =-8\pi G(t)\left(  \omega\rho_{m}+p_{\phi
}\right)  ,\\
\rho_{m}^{\prime}+\left(  \rho_{m}+p_{m}\right)  H  &  =-\frac{G^{\prime}}%
{G}\rho_{m},\\
\phi^{\prime\prime}+H\phi^{\prime}+\frac{dV}{d\phi}  &  =-\frac{G^{\prime}}%
{G}\frac{1}{\phi^{\prime}}\left(  \frac{1}{2}\phi^{\prime2}+V(\phi)\right)  .
\end{align}

With a potential given by Eq. (\ref{Th2}) we have found the next solution%
\begin{align}
a_{2}  &  =\frac{1}{2}\left(  1\pm\sqrt{1-4m^{2}}\right)  \qquad\forall
m\in\left[  -\frac{1}{2},\frac{1}{2}\right]  \backslash\left\{  0\right\}
,\qquad\alpha=\alpha>0,\nonumber\\
\beta &  =\frac{3\omega\left(  1\pm\sqrt{\left(  1-4m^{2}\right)  }\right)
+G_{0}\alpha\left(  1-\omega\right)  +1\pm\sqrt{\left(  1-4m^{2}\right)
}-4m^{2}\left(  1+\omega\right)  }{2G_{0}\left(  \omega+1\right)
},\nonumber\\
\rho_{0}  &  =\frac{1\pm\sqrt{\left(  1-4m^{2}\right)  }-G_{0}\alpha^{2}%
}{G_{0}\left(  \omega+1\right)  },\qquad G_{0}=const.>0,
\end{align}
thus%
\begin{equation}
\rho_{m}=\rho_{0}\left(  t+t_{0}\right)  ^{-\left(  \left(  \omega+1\right)
h+2\alpha\right)  },\qquad p_{m}=\omega\rho_{m},
\end{equation}%
\begin{equation}
a_{1}=a_{0}\left(  t+t_{0}\right)  ,\qquad b=b_{0}\left(  t+t_{0}\right)
^{a_{2}},\qquad d=d_{0}\left(  t+t_{0}\right)  ^{a_{2}},
\end{equation}%
\begin{equation}
\phi=\phi_{0}\left(  t+t_{0}\right)  ^{-\alpha},\quad V=\beta\left(
t+t_{0}\right)  ^{-2\left(  \alpha+1\right)  },\; G=G_{0}\left(
t+t_{0}\right)  ^{2\alpha}.
\end{equation}

Therefore, as in the last model, \ with the above restrictions on the
$m-$parameter our model does not accelerate but isotropizes since
$\mathcal{W}^{2}\rightarrow0$ as well as $\mathcal{A}.$ With regard to the
quantity $P^{2}$ it is observed that $P^{2}\rightarrow0$, $\forall m\in\left(
-\frac{1}{2},\frac{1}{2}\right)  \backslash\left\{  0\right\}  $ and only runs
to minus infinity when $m\rightarrow\pm\frac{1}{2}.$ $G$ is an increasing time
function as in the above model.

\section{Scalar-tensor model}

We consider the following field equations for the BD model \cite{Will},
\begin{equation}
R_{ij}-\frac{1}{2}g_{ij}R=\frac{8\pi}{\phi}T_{ij}^{m}+\Lambda\left(
\phi\right)  g_{ij}+\frac{\omega}{\phi^{2}}\left(  \phi_{,i}\phi_{,j}-\frac
{1}{2}g_{ij}\phi_{,l}\phi^{,l}\right)  +\frac{1}{\phi}\left(  \phi
_{;ij}-g_{ij}\square\phi\right)  ,
\end{equation}

\begin{equation}
\square\phi+\frac{1}{2}\phi_{,l}\phi^{,l}\frac{d}{d\phi}\ln\left(
\frac{\omega\left(  \phi\right)  }{\phi}\right)  +\frac{1}{2}\frac{\phi
}{\omega\left(  \phi\right)  }\left(  R+2\frac{d}{d\phi}\left(  \phi
\Lambda\left(  \phi\right)  \right)  \right)  =0.
\end{equation}

The arbitrary functions $\omega\left(  \phi\right)  $ and $\Lambda\left(
\phi\right)  $ distinguish the different scalar-tensor theories of
gravitation. $\Lambda\left(  \phi\right)  $ is a potential function and plays
the role of a cosmological constant, and $\omega\left(  \phi\right)  $ is the
coupling function of the particular theory. $T_{ij}^{m}$ is the matter
stress-energy tensor.

The last equation can be substituted by%
\begin{equation}
\square\phi+\frac{2}{3+2\omega\left(  \phi\right)  }\left(  \phi^{2}%
\frac{d\Lambda}{d\phi}-\phi\Lambda\left(  \phi\right)  \right)  =\frac
{1}{\left(  3+2\omega\left(  \phi\right)  \right)  }\left(  8\pi
T-\frac{d\omega}{d\phi}\phi_{,l}\phi^{,l}\right)  ,
\end{equation}
where $T=T_{i}^{i}$ is the trace of the stress-energy tensor, where we have
assumed $\phi=\phi(t)$, and the derivatives respect $t$ are denoted by a
comma. Furthermore it is verified the following relationship: $divT=0, $
$i.e.$%
\begin{equation}
\rho^{\prime}+\left(  \rho+p\right)  H=0.
\end{equation}

In what follows we shall assume $\omega\left(  \phi\right)  =const$,
$\Lambda=\Lambda\left(  \phi\right)  $. The corresponding field equations with
a perfect fluid for the matter content in the homogeneous line element
(Bianchi $\mathrm{VI}_{0}$ model) will be calculated. Thus the field equations
are%
\begin{align}
\frac{a^{\prime}}{a}\frac{b^{\prime}}{b}+\frac{a^{\prime}}{a}\frac{d^{\prime}%
}{d}+\frac{d^{\prime}}{d}\frac{b^{\prime}}{b}-\left(  2+\frac{b^{2}}{d^{2}%
}+\frac{d^{2}}{b^{2}}\right)  \frac{m^{2}}{4a^{2}}  &  =\frac{8\pi}{\phi}%
\rho-H\frac{\phi^{\prime}}{\phi}+\frac{\omega}{2}\left(  \frac{\phi^{\prime}%
}{\phi}\right)  ^{2}+\Lambda\left(  \phi\right)  ,\\
\frac{b^{\prime\prime}}{b}+\frac{d^{\prime\prime}}{d}+\frac{d^{\prime}}%
{d}\frac{b^{\prime}}{b}+\left(  2+\frac{b^{2}}{d^{2}}+\frac{d^{2}}{b^{2}%
}\right)  \frac{m^{2}}{4a^{2}}  &  =-\frac{8\pi}{\phi}p-\frac{\phi^{\prime}%
}{\phi}\left(  \frac{d^{\prime}}{d}+\frac{b^{\prime}}{b}\right)  -\frac
{\omega}{2}\left(  \frac{\phi^{\prime}}{\phi}\right)  ^{2}-\frac{\phi
^{\prime\prime}}{\phi}+\Lambda\left(  \phi\right)  ,\\
\frac{a^{\prime\prime}}{a}+\frac{b^{\prime\prime}}{b}+\frac{a^{\prime}}%
{a}\frac{b^{\prime}}{b}-\left(  2+\frac{3d^{2}}{b^{2}}-\frac{b^{2}}{d^{2}%
}\right)  \frac{m^{2}}{4a^{2}}  &  =-\frac{8\pi}{\phi}p-\frac{\phi^{\prime}%
}{\phi}\left(  \frac{a^{\prime}}{a}+\frac{b^{\prime}}{b}\right)  -\frac
{\omega}{2}\left(  \frac{\phi^{\prime}}{\phi}\right)  ^{2}-\frac{\phi
^{\prime\prime}}{\phi}+\Lambda\left(  \phi\right)  -\cosh\left(  mx\right)
^{2}\frac{\phi^{\prime}}{\phi}\left(  \frac{d^{\prime}}{d}-\frac{b^{\prime}%
}{b}\right)  ,\\
\frac{b^{\prime\prime}}{b}-\frac{d^{\prime\prime}}{d}+\frac{a^{\prime}}%
{a}\left(  \frac{b^{\prime}}{b}-\frac{d^{\prime}}{d}\right)  +\frac{m^{2}%
}{a^{2}}\left(  \frac{b^{2}}{d^{2}}-\frac{d^{2}}{b^{2}}\right)   &
=\frac{\phi^{\prime}}{\phi}\left(  \frac{b^{\prime}}{b}-\frac{d^{\prime}}%
{d}\right)  ,\\
\frac{d^{\prime\prime}}{d}+\frac{a^{\prime\prime}}{a}+\frac{a^{\prime}}%
{a}\frac{d^{\prime}}{d}-\left(  2+\frac{3b^{2}}{d^{2}}-\frac{d^{2}}{b^{2}%
}\right)  \frac{m^{2}}{4a^{2}}  &  =-\frac{8\pi}{\phi}p-\frac{\phi^{\prime}%
}{\phi}\left(  \frac{a^{\prime}}{a}+\frac{d^{\prime}}{d}\right)  -\frac
{\omega}{2}\left(  \frac{\phi^{\prime}}{\phi}\right)  ^{2}-\frac{\phi
^{\prime\prime}}{\phi}+\Lambda\left(  \phi\right)  -\cosh\left(  mx\right)
^{2}\frac{\phi^{\prime}}{\phi}\left(  \frac{d^{\prime}}{d}-\frac{b^{\prime}%
}{b}\right)  ,
\end{align}%
\begin{align}
\left(  3+2\omega\left(  \phi\right)  \right)  \left(  \frac{\phi
^{\prime\prime}}{\phi}+H\frac{\phi^{\prime}}{\phi}\right)  -2\left(
\Lambda-\phi\frac{d\Lambda}{d\phi}\right)   &  =\frac{8\pi}{\phi}\left(
\rho-3p\right)  ,\\
\rho^{\prime}+\left(  \rho+p\right)  H  &  =0.
\end{align}

Since we are only interested in finding self-similar solutions then, if we
take into account our previous results, i.e. $b=d$, the FE reads%
\begin{align}
2\frac{a^{\prime}}{a}\frac{b^{\prime}}{b}+\left(  \frac{b^{\prime}}{b}\right)
^{2}-\frac{m^{2}}{a^{2}}  &  =\frac{8\pi}{\phi}\rho-H\frac{\phi^{\prime}}%
{\phi}+\frac{\omega}{2}\left(  \frac{\phi^{\prime}}{\phi}\right)  ^{2}%
+\Lambda\left(  \phi\right)  ,\\
2\frac{b^{\prime\prime}}{b}+\left(  \frac{b^{\prime}}{b}\right)  ^{2}%
+\frac{m^{2}}{a^{2}}  &  =-\frac{8\pi}{\phi}p-2\frac{\phi^{\prime}}{\phi}%
\frac{b^{\prime}}{b}-\frac{\omega}{2}\left(  \frac{\phi^{\prime}}{\phi
}\right)  ^{2}-\frac{\phi^{\prime\prime}}{\phi}+\Lambda\left(  \phi\right)
,\\
\frac{a^{\prime\prime}}{a}+\frac{b^{\prime\prime}}{b}+\frac{a^{\prime}}%
{a}\frac{b^{\prime}}{b}-\frac{m^{2}}{a^{2}}  &  =-\frac{8\pi}{\phi}%
p-\frac{\phi^{\prime}}{\phi}\left(  \frac{a^{\prime}}{a}+\frac{b^{\prime}}%
{b}\right)  -\frac{\omega}{2}\left(  \frac{\phi^{\prime}}{\phi}\right)
^{2}-\frac{\phi^{\prime\prime}}{\phi}+\Lambda\left(  \phi\right)  ,\\
\frac{d^{\prime\prime}}{d}+\frac{a^{\prime\prime}}{a}+\frac{a^{\prime}}%
{a}\frac{d^{\prime}}{d}-\frac{m^{2}}{a^{2}}  &  =-\frac{8\pi}{\phi}%
p-\frac{\phi^{\prime}}{\phi}\left(  \frac{a^{\prime}}{a}+\frac{b^{\prime}}%
{b}\right)  -\frac{\omega}{2}\left(  \frac{\phi^{\prime}}{\phi}\right)
^{2}-\frac{\phi^{\prime\prime}}{\phi}+\Lambda\left(  \phi\right)  ,
\end{align}
and the conservation equations%
\begin{align}
\left(  3+2\omega\left(  \phi\right)  \right)  \left(  \frac{\phi
^{\prime\prime}}{\phi}+H\frac{\phi^{\prime}}{\phi}\right)  -2\left(
\Lambda-\phi\frac{d\Lambda}{d\phi}\right)   &  =\frac{8\pi}{\phi}\left(
\rho-3p\right)  ,\\
\rho^{\prime}+\left(  \rho+p\right)  H  &  =0,\qquad\Longleftrightarrow
\qquad\rho=\rho_{0}t^{-\alpha},
\end{align}
where $H=h\left(  t+t_{0}\right)  ^{-1},$ with $h=\left(  1+2a_{2}\right)  ,$
and $\alpha=h\left(  1+\gamma\right)  $, we are taking into account the
equation of state $p=\gamma\rho,$ $\gamma\in(-1,1].$

In order to solve the resulting FE we need to integrate%
\begin{equation}
\left(  3+2\omega\right)  \left(  \frac{\phi^{\prime\prime}}{\phi}+\frac{h}%
{t}\frac{\phi^{\prime}}{\phi}\right)  -2\left(  \Lambda-\phi\frac{d\Lambda
}{d\phi}\right)  =\frac{8\pi}{\phi}\frac{\left(  1-3\gamma\right)  \rho_{0}%
}{t^{\alpha}}, \label{Laurita}%
\end{equation}

We may study this equation through the LG method, i.e. we study the kind of
functions $\Lambda\left(  \phi\right)  $ such that this equation is integrable
in a closed form. We start by rewriting it in an appropriate way%
\begin{equation}
\phi^{\prime\prime}+ht^{-1}\phi^{\prime}-B\left(  \Lambda-\phi\frac{d\Lambda
}{d\phi}\right)  \phi-Ct^{-\alpha}=0, \label{kata2}%
\end{equation}
where%
\begin{equation}
h=\left(  1+2a_{2}\right)  ,\qquad B=\frac{2}{\left(  3+2\omega\right)
},\qquad C=\frac{8\pi\left(  1-3\gamma\right)  \rho_{0}}{\left(
3+2\omega\right)  },
\end{equation}
and we shall deduce that $\alpha=2-n.$

We need to solve the following system of PDEs%
\begin{align}
t^{2}\xi_{\phi\phi}  &  =0,\\
2ht^{-1}\xi_{\phi}+\eta_{\phi\phi}-2\xi_{\phi t}  &  =0,\\
-3\left(  B\phi\left(  \Lambda-\phi\Lambda_{\phi}\right)  +Ct^{-\alpha
}\right)  \xi_{\phi}+ht^{-2}\left(  t\xi_{t}-\xi\right)  +2\eta_{t\phi}%
-\xi_{tt}  &  =0, \label{kata6}%
\end{align}%
\begin{equation}
B\eta\left(  \phi^{2}\Lambda_{\phi\phi}-\left(  \Lambda-\phi\Lambda_{\phi
}\right)  \right)  -2B\phi\left(  \Lambda-\phi\Lambda_{\phi}\right)  \xi
_{t}+\frac{C}{t^{-\alpha}}\left(  \alpha\frac{\xi}{t}-2\xi_{t}\right)
+\left(  B\phi\left(  \Lambda-\phi\Lambda_{\phi}\right)  +C\frac{C}%
{t^{-\alpha}}\right)  \eta_{\phi}+\frac{h}{t}\eta_{t}+\eta_{tt}=0.
\label{kata7}%
\end{equation}

For example, if we impose the symmetry
\begin{equation}
\xi=t,\qquad\eta=n\phi, \label{kata8}%
\end{equation}
brings us to obtain the following restriction on $\Lambda\left(  \phi\right)
.$ From Eq. (\ref{kata7}) we get%
\begin{equation}
B\phi\left(  n\phi^{2}\Lambda_{\phi\phi}-2\left(  \Lambda-\phi\Lambda_{\phi
}\right)  \right)  +Ct^{-\alpha}\left(  \alpha-2+n\right)  =0,
\end{equation}
and therefore, we find as result, that:%
\begin{equation}
n=2-\alpha,
\end{equation}
and%
\begin{equation}
n\phi^{2}\Lambda_{\phi\phi}-2\left(  \Lambda-\phi\Lambda_{\phi}\right)
=0\qquad\Longrightarrow\qquad\Lambda=\phi^{-2/n},
\end{equation}
is a solution. In fact the most general solution is $\Lambda=C_{1}\phi
+C_{2}\phi^{-2/n}.$ This result is valid for all the self-similar Bianchi
models. Furthermore, the symmetry Eq. (\ref{kata8}) brings us to obtain a
particular solution of Eq. (\ref{kata2}) which is given by, $\phi=\phi
_{0}t^{n},$ with $\phi_{0}\in\mathbb{R},$ with $n=2-\alpha,$ and the following
constrains on the numerical constants must be verified: $n(n-1)+hn-2(1+\frac
{2}{n})B-C=0,$ $\alpha=h\left(  1+\gamma\right)  =\left(  1+2a_{2}\right)
\left(  1+\gamma\right)  .$ We would like to emphasize that other solutions
could be obtained with this procedure by imposing other symmetries.

We find the following solution:%
\begin{align}
\phi_{0}  &  =\phi_{0}\qquad\phi_{0}=1,\qquad\qquad\Lambda_{0}=0,\nonumber\\
a_{2}  &  =-\frac{\left(  \gamma-1\right)  \left(  \omega\left(
\gamma-1\right)  -1\right)  }{2\omega\left(  \gamma^{2}-1\right)  +\gamma
-3},\qquad a_{2}=0,\quad\Longleftrightarrow\quad\gamma=1,\nonumber\\
q  &  =\frac{2\left(  \omega\left(  \left(  3\gamma+1\right)  \left(
\gamma-1\right)  \right)  -2\right)  }{4\omega\left(  \gamma-1\right)
+3\gamma-5},\quad q=0,\quad\Longleftrightarrow\quad\gamma=A_{\pm},\nonumber\\
\rho_{0}  &  =-\frac{\left(  3+2\omega\right)  ^{2}\left(  \gamma-1\right)
^{3}}{8\pi\left(  2\omega\left(  \gamma^{2}-1\right)  +\gamma-3\right)
},\quad\rho_{0}=0,\quad\Longleftrightarrow\quad\gamma=1,\nonumber\\
m  &  =\frac{\left(  \gamma-1\right)  \sqrt{-2\left(  3+2\omega\right)
\left(  \omega\left(  \left(  3\gamma+1\right)  \left(  \gamma-1\right)
\right)  -2\right)  }}{2\left(  2\omega\left(  \gamma^{2}-1\right)
+\gamma-3\right)  },\nonumber\\
m  &  =0,\qquad\Longleftrightarrow\qquad\gamma=1\wedge A_{\pm},\nonumber\\
n  &  =\frac{-\left(  3\gamma-1\right)  \left(  \gamma-1\right)  }%
{2\omega\left(  \gamma^{2}-1\right)  +\gamma-3},\quad n=0,\quad
\Longleftrightarrow\quad\gamma=1\wedge\frac{1}{3},
\end{align}
where $A_{\pm}=\frac{\left(  \omega\pm\sqrt{2\omega\left(  2\omega+3\right)
}\right)  }{3\omega}.$

We get for the BD parameter $\omega$ that according to solar system
experiments is $\omega\thickapprox500.$ (see \cite{46}). A better estimation
of this parameter should be obtained from measure of other cosmological
parameters in order to constrain $\omega$ more strongly than by means of solar
system experiments (see \cite{47}). However, theories of the very early
Universe such as string theory, are better described in the context of JBD,
which shows that $\omega$ can take negative values (see for example
\cite{48}). A recent value for $\omega$ is $\omega\thickapprox3300$
\cite{Faraoni}.

If we fix $\omega=3300,$ then we get the following results%
\begin{align}
A_{+}  &  =1.0002,\quad A_{-}=-0.33348,\quad m>0,\quad\forall\gamma\in
(A_{-},1],\nonumber\\
a_{2}  &  \geq0,\quad\forall\gamma\in(-1,1],\qquad q=\left\{
\begin{array}
[c]{cc}%
<0 & \forall\gamma<A_{-}\\
=0 & \gamma=A_{-}\\
>0 & \forall\gamma>A_{-}%
\end{array}
\right.  ,\nonumber\\
\rho_{0}  &  \geq0,\quad\forall\gamma\in(-1,1],\quad n=\left\{
\begin{array}
[c]{cc}%
>0 & \forall\gamma<1/3\\
=0 & \gamma=1/3\wedge1\\
<0 & \forall\gamma\in\left(  \frac{1}{3},1\right)
\end{array}
\right.  .
\end{align}
Therefore this solution, with $\omega=3300$, is only valid $\forall\gamma
\in(A_{-},1].$ Note that if $\gamma<A_{-}$ then $m$ is not defined. This means
that%
\begin{align}
a  &  =a_{0}\left(  t+t_{0}\right)  ,\;b=b_{0}\left(  t+t_{0}\right)  ^{a_{2}%
},\;d=d_{0}\left(  t+t_{0}\right)  ^{a_{2}},\;H=h\left(  t+t_{0}\right)
^{-1},\nonumber\\
\rho &  =\rho_{0}\left(  t+t_{0}\right)  ^{-\alpha},\quad\phi=\phi_{0}\left(
t+t_{0}\right)  ^{n},\quad\Lambda=0,
\end{align}
and therefore the scale factors are increasing time functions, the energy
density is a positive time decreasing function. The solution does not inflate
since $q>0$ $\forall\gamma\in(A_{-},1].$ $\phi$ is a positive growing time
function if $\gamma\in(A_{-},1/3)$, it is constant if $\gamma=1/3$ and it
behaves as a decreasing time function if $\gamma\in(1/3,1),$ if $\gamma=1$
then it behaves as a constant. This means that $G$ is a decreasing time
function if $\gamma\in(A_{-},1/3),$ it behaves as a true constant if
$\gamma=1/3\wedge1$ and if $\gamma\in(1/3,1)$ then it is a growing time
function. The cosmological \textquotedblleft constant\textquotedblright,
$\Lambda$ vanishes. Notice that $\alpha=\frac{\left(  4\omega\left(
\gamma-1\right)  +3\gamma-5\right)  (\gamma+1)}{2\omega\left(  \gamma
^{2}-1\right)  +\gamma-3}.$

In order to check if the solution isotropize we compute the quantities
$\mathcal{A}$ and $\mathcal{W}^{2}$, obtaining
\begin{align}
\mathcal{A}  &  =\frac{\left(  \omega\left(  \gamma^{2}-2\gamma-1\right)
-2\right)  ^{2}}{\left(  4\omega\left(  \gamma-1\right)  +3\gamma-5\right)
^{2}}=const\in\left(  0,1\right)  ,\\
\mathcal{W}^{2}  &  =-\frac{\left(  \gamma-1\right)  ^{2}\left(  \omega\left(
\gamma^{2}-2\gamma-1\right)  -2\right)  ^{2}}{36\left(  4\omega\left(
\gamma-1\right)  +3\gamma-5\right)  ^{4}}\left(  8\omega\left(  \left(
2\gamma+1\right)  \left(  \gamma-1\right)  \right)  -3\gamma^{2}%
+6\gamma-15\right)  ,
\end{align}
$\mathcal{W}^{2}=const\in\left(  0,0.01\right)  $, $\forall\gamma\in(A_{-},1]$
and $\omega=3300.$ $\mathcal{A}(A_{-})=\mathcal{W}^{2}(A_{-})=0.$ With regard
to the gravitational entropy we have obtained the following behaviour:
$P^{2}(A_{-})=0,$ if $\gamma\in\left(  A_{-},-0.2000969530\right)  $ then
$P^{2}>0,$ $P^{2}(-0.2000969530)=0$ and it runs to $-\infty$ $\forall\gamma
\in(-0.2000969530,1]$, so once again, we have checked that this quantity is
not a good definition for gravitational entropy.

\subsection{The particular case $\gamma=1/3.$}

In this case $T=0$ and therefore we find the next solution%
\begin{align}
\phi_{0}  &  =\phi_{0},\quad\phi_{0}=1,\nonumber\\
\Lambda_{0}  &  =\frac{1}{6}\left[  \left(  3+2\omega\right)  \left(
4m+\sqrt{3}\right)  ^{2}\right]  ,\quad\Lambda_{0}=0\;\Longleftrightarrow
\;m=-\frac{\sqrt{3}}{4},\nonumber\\
a_{2}  &  =\frac{\sqrt{3}}{3}\left(  3m+\sqrt{3}\right)  ,\quad a_{2}%
=0\qquad\Longleftrightarrow\quad m=-\frac{\sqrt{3}}{3},\nonumber\\
q  &  =\frac{-2m}{2m+\sqrt{3}},\quad q=\left\{
\begin{array}
[c]{cc}%
>0 & \forall m<0\\
<0 & \forall m>0
\end{array}
\right.  ,\nonumber\\
\rho_{0}  &  =-\frac{1}{16\pi}\left(  32\omega\left(  m+\frac{1}{4}\sqrt
{3}\right)  ^{2}+9+24m\sqrt{3}+44m^{2}\right)  ,\nonumber\\
\alpha &  =\frac{4\sqrt{3}}{3}\left(  2m+\sqrt{3}\right)  ,\nonumber\\
\rho_{0}  &  =0\quad\Longleftrightarrow\quad m=-\frac{\sqrt{3}\left(
4\omega\pm\sqrt{2\omega+3}+6\right)  }{16\omega+22},\nonumber\\
n  &  =-\frac{2\sqrt{3}}{3}\left(  4m+\sqrt{3}\right)  ,\quad n=0\quad
\Longleftrightarrow\quad m=-\frac{\sqrt{3}}{4}.
\end{align}
If we fix $\omega=3300,$ then we get the following results%
\begin{align}
\rho_{0}  &  >0,\qquad\forall m\in I,\qquad\alpha>0,\qquad\forall m\in
I,\nonumber\\
\Lambda_{0}  &  \geq0,\;\forall m\in I,\;\Lambda_{0}=0\;\Longleftrightarrow
\;m_{\Lambda_{0}}=-\frac{\sqrt{3}}{4}=-0.43301,\nonumber\\
a_{2}  &  >0,\;\forall m\in I,\;a_{2}=0\;\Longleftrightarrow\;m=-\frac
{\sqrt{3}}{3}\notin I,\nonumber\\
q  &  >0,\;\forall m\in I,\nonumber\\
n  &  =\left\{
\begin{array}
[c]{cc}%
>0 & \forall m\in\left(  -0.43569,m_{\Lambda_{0}}\right) \\
=0 & m=m_{\Lambda_{0}}\\
<0 & \forall m\in\left(  m_{\Lambda_{0}},-0.43036\right)
\end{array}
\right.  ,
\end{align}
therefore this solution, with $\omega=3300$, is only valid $\forall m\in I,$
where $I=\left(  -0.43569,-0.43036\right)  .$ Note that $m_{\Lambda_{0}}\in
I,$
\begin{equation}
\Lambda_{0}\left(  m_{\Lambda_{0}}\right)  =0,\;a_{2}\left(  m_{\Lambda_{0}%
}\right)  =\frac{1}{4},\;q\left(  m_{\Lambda_{0}}\right)  =1,\;\rho_{0}\left(
m_{\Lambda_{0}}\right)  =1.4921\times10^{-2}.
\end{equation}

This means that%
\begin{align}
a  &  =a_{0}\left(  t+t_{0}\right)  ,\;b=b_{0}\left(  t+t_{0}\right)  ^{a_{2}%
},\;d=d_{0}\left(  t+t_{0}\right)  ^{a_{2}},\;H=h\left(  t+t_{0}\right)
^{-1},\nonumber\\
\rho &  =\rho_{0}\left(  t+t_{0}\right)  ^{-\alpha},\;\phi=\phi_{0}\left(
t+t_{0}\right)  ^{n},\;\Lambda=\Lambda_{0}\left(  t+t_{0}\right)  ^{-2},
\end{align}
and therefore the scale factors are growing time functions, the energy density
is a positive time decreasing function. The solution does not inflate since
$q>0$ $\forall m\in I.$ $\phi$ is a positive growing time function if
$m\in\left(  -0.43569,m_{\Lambda_{0}}\right)  $, it is constant if
$m=m_{\Lambda_{0}}$ and it behaves as a decreasing time function if
$m\in\left(  m_{\Lambda_{0}},-0.43036\right)  $. This means that $G$ is a
decreasing time function if $m\in\left(  -0.43569,m_{\Lambda_{0}}\right)  ,$
it behaves as a true constant if $m=m_{\Lambda_{0}}$ and if $m\in\left(
-0.43569,m_{\Lambda_{0}}\right)  $ then it is a growing time function. The
cosmological \textquotedblleft constant\textquotedblright, $\Lambda$ is a
positive decreasing time function except in $m=m_{\Lambda_{0}}$.

\section{Conclusions}

In this paper we have studied some Bianchi types $\mathrm{VI}_{0}$ (with an
unusual metric) models under the self-similarity hypothesis. We have started
by comparing our results with the \textquotedblleft
classical\textquotedblright\ perfect fluid solution already studied by
Collins, Wainwright and Hsu and other authors \cite{Tony1}. Furthermore, we
have been able to improve the solutions since we have found a non-singular
solution for the scale factors i.e. they behave as $a(t)\sim\left(
t+t_{0}\right)  ^{a_{1}}$. However, the metric employed in this paper is very
restrictive, since it allows us to obtain less solutions than with the usual
one \cite{Tony1}. Nevertheless we have been able to obtain a new solution for
the case of a perfect fluid with time-varying constants. This solution is not
inflationary but it is very close to isotropizing since the quantities
$\mathcal{A}$ and $\mathcal{W}^{2}$ take values very close to zero. In fact,
for an adequate selection of the parameters $\omega$ and $m$ they run to zero.
This solution is valid for all $\omega\in(-1,1]$ and $m\in\left[
-1/2,1/2\right]  .$ In this case we have been able to enlarge the range of
validity for the equation of state and we have shown that if $G$ behaves as a
growing time function then $\Lambda$ is a \textquotedblleft\emph{positive}%
\textquotedblright\ decreasing time function. In the same way, if $G$ is
decreasing then $\Lambda$ behaves as a \textquotedblleft\emph{negative}%
\textquotedblright\ decreasing time function. With regard to the gravitational
entropy, we have come to the conclusion that the quantity $P^{2}$ is not an
acceptable candidate for gravitational entropy along the homothetic
trajectories of any self-similar spacetime (in all the cases studied in this paper).

In the second model we have studied a scalar field. We have started this
section by calculating the potentials compatible with the self-similar
solutions. Inversely, we have proved that for such scalar fields the scale
factor must follow a power law solution. These theorems are very general and
are valid for all Bianchi models. We have studied two cases. From the first
one, with a scalar field alone, we have obtained a solution that is not
inflationary but it could be considered to be very close to isotropize, since,
as above, since the quantities $\mathcal{A}$ and $\mathcal{W}^{2}$ take values
very close to zero. In the second case, we have studied a non-interacting
scalar and matter fields. The solution is not inflationary but isotropize as
in the previous cases, and it is valid $\forall\omega\in(-1/3,1)$ and
$m\in(0,1/2].$ In order to incorporate into this framework a variable $G,$ we
have proposed, in a phenomenological way, a modified Klein-Gordon equation. We
have studied the kind of potential compatible with a self-similar solution and
a variable $G$. Once we have deduced the potential and the scalar field then
we study two cases, a scalar field with a $G-$var and a scalar field with a
matter field. The conservation equation outlined  in this case is quite
similar to the one employed in the case of interacting scalar fields. The
solutions obtained are similar since the scale factor is the same and
therefore they are not inflationists and close to isotropize. In both cases,
$G$ behaves as a positive increasing time function.

In the scalar-tensor model, for simplicity, we have chosen, $\omega\left(
\phi\right)  =const.$ and $\Lambda\left(  \phi\right)  $ playing the role of
an effective cosmological constant. As we have shown, the resulting FE are
quite difficult to study. Nevertheless, since we are only interested in
studying self-similar solutions, we have been able to simplify the FE. We
would like to stress that we have not needed to make any assumption in order
to integrate them. By using the Lie group method we have obtained a possible
form for the dynamical cosmological constant, $\Lambda\left(  \phi\right)
=\phi^{-2/n}=t^{-2},$ since $\phi=\phi_{0}t^{n}.$ In the same way, we
emphasize that this result is valid for all the self-similar Bianchi models.
With this result we have obtained several solutions for the model. We have
considered that the first of the obtained solutions is unphysical since
$\rho_{0}<0$ if the coupling parameter $\omega$ is positive as the recent
observations suggest. Nevertheless if we consider $\omega<0$, as it is
suggested by the string theories \cite{48}, this solution has physical
meaning. In the second solution, $\Lambda=0$ and it is only valid when
$\gamma\in I,$ where $I=(A_{-},1],$ where $A_{-}=-0.33348$ if $\omega=3300$ as
recent experiments suggest \cite{Faraoni}. For such values $\phi$ behaves as a
growing time function if $\gamma\in\left(  A_{-},1/3\right)  ,$ it is constant
if $\gamma=1/3\wedge1$ and it is a positive decreasing time function in the
interval $\gamma\in\left(  1/3,1\right)  .$ Therefore, $G$ is decreasing,
constant and growing in the same intervals. In the same way we have found that
this solution does not inflate since $q>0,$ $\gamma\in I,$ which is unusual.
We may also say that this solution would be considered to be very close to an
isotropy state since the Weyl parameter, $\mathcal{W}^{2}\ll1.$ In fact this
quantity takes values very close to zero $\gamma\in I.$ We have also studied
the particular solution $\gamma=1/3.$ In this case the trace of the
stress-energy tensor vanishes and therefore the conservation equation reduces
to a very simple ODE. This solution is only valid for a very restrictive
interval, $m\in I_{m}=\left(  m_{1},m_{2}\right)  ,$ where $\rho_{0}>0.$ We
have found that in this case $\Lambda$ behaves as a positive decreasing time
function except when $m=m_{\Lambda_{0}}\in I,$ for which, $\Lambda\left(
m_{\Lambda_{0}}\right)  =0.$ In the same way, we have found that $\phi$ is a
growing time function when $m\in\left(  m_{1},m_{\Lambda_{0}}\right)  ,$ it is
constant if $m=m_{\Lambda_{0}}$ and it behaves as a decreasing time function
if $m\in\left(  m_{\Lambda_{0}},m_{2}\right)  .$ $G$ behaves in the inverse
way in the same intervals. This solution does not inflate as the one above
does. Finally we would like to stress the fact that, in this solution, the
exponents, $a_{2},$ $n$ and $\alpha$ are irrational numbers. This could have
implications for the integrability of the FE (see \cite{Yoshida1,Yoshida2}).

\appendix

\section{Matter collineation for the scalar model}

In this section we shall study the matter collineations for the scalar field
following a method employed in \cite{Tony2}. We start by defining a generic
vector field $X=\left(  X_{i}\left(  t,x,y,z\right)  \right)  _{i=1}^{4}%
\in\mathfrak{X}(M).$ The energy-momentum tensor is defined by Eq.
(\ref{PF-tensor}). The metric tensor $g_{ij}$ is defined by Eq.
(\ref{BVIo-metric}). In recent years, much interest has been shown in the
study of matter collineation (MCs) (see for example \cite{Sharif}-\cite{TA}).
A vector field along which the Lie derivative of the energy-momentum tensor
vanishes is called an MC, i.e. ${\mathcal{L}}_{X}T_{ij}=0,$where $X^{i}$ is
the symmetry or collineation vector. Also, assuming the Einstein field
equations, a vector $X^{i}$ generates an MC if ${\mathcal{L}}_{X}G_{ij}=0$. It
is obvious that the symmetries of the metric tensor (isometries) are also
symmetries of the Einstein tensor $G_{ij}$, but this is not necessarily the
case for the symmetries of the Ricci tensor (Ricci collineations) which are
not, in general, symmetries of the Einstein tensor. If $X$ is a Killing vector
(KV) (or a homothetic vector), then ${\mathcal{L}}_{X}T_{ij}=0$, thus every
isometry is also an MC but the converse is not true, in general. Notice that
collineations can be proper (non-trivial) or improper (trivial). A proper MC
is defined to be an MC which is not a KV, or a homothetic vector. Carot et al
(see \cite{ccv}) and Hall et al.(see \cite{hrv}) have noticed some important
general results about the Lie algebra of MCs. Let $M$ be a spacetime manifold.
Then, generically, any vector field $X\in\mathfrak{X}(M) $ which
simultaneously satisfies ${\mathcal{L}}_{X}T_{ab}=0$ ($\Leftrightarrow
{\mathcal{L}}_{X}G_{ab}=0$) and ${\mathcal{L}}_{X}C_{bcd}^{a}=0$ is a
homothetic vector field i.e. $\mathcal{L}_{X}g=2g$.

The usual matter collineation equations read, ${\mathcal{L}}_{X}T_{ij}^{\phi
}=0,$ $T_{ij}^{\phi}$ is given by Eq. (\ref{STensor}), finding, in this case,
that the obtained matter collineation (MC) is:%
\begin{equation}
X=X_{1}\partial_{t}+\left(  X_{1}^{\prime}-X_{1}H_{2}\right)  y\partial
_{y}+\left(  X_{1}^{\prime}-X_{1}H_{2}\right)  z\partial_{z}, \label{mc}%
\end{equation}
which is a proper MC, and where, as it is observed, if $X_{1}=\left(
t+t_{0}\right)  $, then it is regained the usual homothetic vector field (see
Eq. (\ref{HO})) i.e. a improper MC.

For this reason we may also check that the homothetic vector field also
verifies the equation, ${\mathcal{L}}_{HO}T_{ij}=0,$ (note that it is verified
${\mathcal{L}}_{HO}C_{bcd}^{a}=0$) that we may develop as follows:%
\begin{align}
\rho^{\prime}t+2\rho=0,\nonumber\\
aa^{\prime}+taa^{\prime\prime}-t\left(  a^{\prime}\right)  ^{2} =0,\nonumber\\
g_{33}y\left(  \frac{b^{\prime}}{b}+t\frac{b^{\prime\prime}}{b}-t\frac
{b^{\prime2}}{b^{2}}\right)  +g_{34}z\left(  \frac{d^{\prime}}{d}%
+t\frac{d^{\prime\prime}}{d}-t\frac{d^{\prime2}}{d^{2}}\right)  =0,\nonumber\\
g_{34}y\left(  \frac{b^{\prime}}{b}+t\frac{b^{\prime\prime}}{b}-t\frac
{b^{\prime2}}{b^{2}}\right)  +g_{44}z\left(  \frac{d^{\prime}}{d}%
+t\frac{d^{\prime\prime}}{d}-t\frac{d^{\prime2}}{d^{2}}\right)  =0,\nonumber\\
g_{22}\left(  tp^{\prime}+2p-2pt\frac{a^{\prime}}{a}\right)  +tpg_{22}%
^{\prime} =0,\nonumber\\
g_{33}\left(  tp^{\prime}+2p-2pt\frac{b^{\prime}}{b}\right)  +tpg_{33}%
^{\prime}+px\partial_{x}g_{33}\left(  1-t\frac{a^{\prime}}{a}\right)
=0,\nonumber\\
g_{34}\left(  tp^{\prime}+2p-2pt\left(  \frac{b^{\prime}}{b}+\frac{d^{\prime}%
}{d}\right)  \right)  +tpg_{34}^{\prime}+px\partial_{x}g_{34}\left(
1-t\frac{a^{\prime}}{a}\right)  =0,\nonumber\\
g_{44}\left(  tp^{\prime}+2p-2pt\frac{d^{\prime}}{d}\right)  +tpg_{44}%
^{\prime}+px\partial_{x}g_{44}\left(  1-t\frac{a^{\prime}}{a}\right)  =0.
\end{align}

As we can see, actually, the only ODEs that must by satisfied are:%
\begin{equation}
\rho^{\prime}t+2\rho=0,\qquad\left(  p^{\prime}t+2p\right)  =0,
\end{equation}
which are equivalent. Hence%
\begin{equation}
\rho^{\prime}t=-2\rho,\qquad L_{H}\rho=\rho^{\prime}t=-2\rho,
\end{equation}
i.e.%
\begin{equation}
\left(  \phi^{\prime\prime}\phi^{\prime}+\frac{dV(\phi)}{d\phi}\phi^{\prime
}\right)  t+2\left(  \frac{1}{2}\phi^{\prime2}+V(\phi)\right)  =0,
\label{A1ODE}%
\end{equation}
that we may split into%
\begin{equation}
t\left(  \phi^{\prime\prime}\phi^{\prime}\right)  +\phi^{\prime2}=0,\qquad
t\frac{dV}{d\phi}\phi^{\prime}+2V=0,
\end{equation}
where%
\begin{equation}
t\phi^{\prime\prime}+\phi^{\prime}=0\qquad\Longrightarrow\qquad\phi=\kappa\ln
t, \label{A1phi}%
\end{equation}
i.e. $\mathcal{L}_{H}\phi^{\prime}=0.$ With regard to the second equation%
\begin{equation}
\frac{dV}{d\phi}\kappa+2V=0\qquad\Longrightarrow\qquad V=Ke^{-\frac{2}{\kappa
}\phi}, \label{A1V}%
\end{equation}
i.e. $\mathcal{L}_{H}V=-2V$. Note that from Eq. (\ref{A1phi}), $\phi^{\prime
}t=\kappa.$

We also may study the complete equation (\ref{A1ODE}) i.e.%
\begin{equation}
\phi^{\prime\prime}+\frac{dV}{d\phi}+\phi^{\prime}t^{-1}+2V\left(
t\phi^{\prime}\right)  ^{-1}=0, \label{new_ODE}%
\end{equation}
under the LG method. The standard procedure brings us to get the next system
of PDE:%
\begin{align}
t^{2}\xi_{\phi\phi} =0,\\
t^{2}\eta_{\phi\phi}-2t^{2}\xi_{t\phi}+2t\xi_{\phi} =0,\\
2t^{2}\eta_{t\phi}-t^{2}\xi_{tt}+t\xi_{t}+3t^{2}\xi_{\phi}\frac{dV}{d\phi}%
-\xi=0,\\
t^{2}\eta_{tt}+8t\xi_{\phi}V+2t^{2}\xi_{t}\frac{dV}{d\phi}-2t^{2}\eta_{\phi
}\frac{dV}{d\phi}+t\eta_{t}+\eta t^{2}\frac{d^{2}V}{d\phi^{2}}
=0,\label{A1sys1}\\
-2\xi V+6t\xi_{t}V+2t\eta\frac{dV}{d\phi}-4t\eta_{\phi}V =0,\label{A1sys2}\\
t\eta_{t}V =0.
\end{align}

The symmetry, $\xi=\alpha t,\eta=\delta,$ brings us to obtain the following
restriction on the potential (from Eq. (\ref{A1sys1}) and (\ref{A1sys2}))%
\begin{equation}
2\frac{dV}{d\phi}+\frac{d^{2}V}{d\phi^{2}}=0,\qquad2V+\frac{dV}{d\phi}=0,
\end{equation}
and therefore we obtain as solution%
\begin{equation}
V=\exp\left(  -2\phi\right)  \qquad\phi=\ln t.
\end{equation}

Therefore we may state the following theorem.

\begin{theorem}
The only possible form for the potential $V\left(  \phi\right)  $ for a
spacetime admitting a HFV, $HO$ is $V(\phi)=V_{0}\exp\left(  \kappa
\phi\right)  $ and therefore $\phi=\ln t.$
\end{theorem}

Sometimes it is interesting to study the symmetries of the tensor $T_{i}%
^{j}\in\mathcal{T}_{1}^{1}(M)$. In this case the matter collineation equations
read $\mathcal{L}_{HO}T_{i}^{j}=0,$ iff $\rho^{\prime}=0,$ and $p^{\prime}=0,$
which is equivalent to%
\begin{equation}
\phi^{\prime\prime}=\pm\frac{dV(\phi)}{d\phi}, \label{ODE3}%
\end{equation}
where as we can see this approach is related with the variational symmetries.

In this case the solution of Eq. (\ref{ODE3}) is the following one
\begin{equation}
t=\int^{\phi}\pm\frac{da}{\sqrt{-2V(a)+C_{1}}}+C_{2}.
\end{equation}

The Lie group methods applied to Eq. (\ref{ODE3}) gives%
\begin{align}
\xi_{\phi\phi} =0,\\
\eta_{\phi\phi}-2\xi_{t\phi} =0,\\
2\eta_{t\phi}-\xi_{tt}+3\xi_{\phi}\frac{dV}{d\phi} =0,\\
\eta_{tt}+2\xi_{t}\frac{dV}{d\phi}-\eta_{\phi}\frac{dV}{d\phi}+\eta\frac
{d^{2}V}{d\phi^{2}} =0, \label{A1sys3}%
\end{align}
where, for example, the symmetry $\xi=t,\eta=1$ brings us to obtain, from Eq.
(\ref{A1sys3}), the following restriction on the potential $V$%
\begin{equation}
2\frac{dV}{d\phi}+\frac{d^{2}V}{d\phi^{2}}=0,
\end{equation}
i.e. a solution like this: $V=\exp\left(  -2\phi\right)  ,$ and therefore,
$\phi=\ln t.$

\section{Matter collineation for the scalar model with $G(t)$}

If the stress-energy tensor stand for a scalar model then it takes the
following form:
\begin{equation}
T_{ij}=\left(  \rho+p\right)  u_{i}u_{j}+pg_{ij},
\end{equation}
where
\[
\rho=\frac{1}{2}\phi^{\prime2}+V(\phi),\qquad p=\frac{1}{2}\phi^{\prime
2}-V(\phi).
\]

Then Eq.%
\[
\mathcal{L}_{HO}\left(  G(t)T_{ij}\right)  =0,
\]
reads
\[
\frac{G^{\prime}}{G}+\frac{\rho^{\prime}}{\rho}=-\frac{2}{t}\qquad
\Longleftrightarrow\qquad G\rho\thickapprox t^{-2},
\]
now reads%
\[
\frac{\rho^{\prime}}{\rho}=-\left(  \frac{2}{t}+\frac{G^{\prime}}{G}\right)
,
\]
i.e.%
\[
\phi^{\prime\prime}=-\frac{dV(\phi)}{d\phi}-\left(  \frac{2}{t}+\frac
{G^{\prime}}{G}\right)  \left(  \frac{1}{2}\phi^{\prime}+\frac{V(\phi)}%
{\phi^{\prime}}\right)  .
\]

We may follow different tactics. The first one consists in studying the whole
equation%
\begin{equation}
\phi^{\prime\prime}=-\frac{dV(\phi)}{d\phi}-\left(  \frac{2}{t}+\frac
{G^{\prime}}{G}\right)  \left(  \frac{1}{2}\phi^{\prime}+\frac{V(\phi)}%
{\phi^{\prime}}\right)  . \label{ODE}%
\end{equation}

The second one will consist in splitting the ODE in the following form (as in
the standard case)%
\begin{align}
\phi^{\prime\prime}  &  =-\left(  \frac{1}{t}+\frac{G^{\prime}}{2G}\right)
\phi^{\prime},\label{mia1}\\
\frac{dV(\phi)}{d\phi}  &  =-\left(  \frac{2}{t}+\frac{G^{\prime}}{G}\right)
\frac{V(\phi)}{\phi^{\prime}}, \label{mia2}%
\end{align}
in such a way that solving (\ref{mia1}) then we will be able to integrate
(\ref{mia2}).

Eq. (\ref{mia1}) has the following solution%
\begin{equation}
\phi=C_{1}\int\frac{dt}{t\sqrt{G(t)}}+C_{2}.
\end{equation}

In the same way Eq. (\ref{mia1}) admits the following symmetries:%
\begin{align*}
\xi_{\phi\phi}  &  =0,\\
2\left(  \frac{1}{t}+\frac{G^{\prime}}{2G}\right)  \xi_{\phi}+\eta_{\phi\phi
}-2\xi_{t\phi}  &  =0,\\
\left(  \frac{1}{t}+\frac{G^{\prime}}{2G}\right)  \xi_{t}+\left(  -\frac
{1}{t^{2}}+\frac{G^{\prime\prime}}{2G}-\frac{G^{\prime2}}{2G^{2}}\right)
\xi+\eta_{t\phi}-2\xi_{tt}  &  =0,\\
\left(  \frac{1}{t}+\frac{G^{\prime}}{2G}\right)  \eta_{t}+\eta_{tt}  &  =0,
\end{align*}
where the symmetry
\begin{equation}
\xi=t,\qquad\eta=-b\phi\qquad\Longrightarrow\qquad\phi=\phi_{0}t^{-b},
\label{sym_f2}%
\end{equation}
brings us to obtain the following constrain on function $G(t):$%
\begin{equation}
G^{\prime\prime}=\frac{G^{\prime2}}{G}-\frac{G^{\prime}}{t}, \label{f2}%
\end{equation}
whose solution is%
\begin{equation}
G=G_{0}t^{k},\qquad k\in\mathbb{R}. \label{sol_f2}%
\end{equation}

From
\begin{equation}
G\rho\thickapprox t^{-2}\qquad\Longrightarrow\qquad G\phi^{\prime
2}\thickapprox t^{-2}\qquad\Longleftrightarrow\qquad k=2b.
\end{equation}

Now, Eq. (\ref{mia2}) yields%
\[
\frac{dV(\phi)}{d\phi}\frac{\phi^{\prime}}{V(\phi)}=-2\left(  1+b\right)
t^{-1},
\]
whose integration gives%
\[
\ln V=-2\left(  1+b\right)  \ln t\qquad\Longleftrightarrow\qquad V=t^{-2b-2},
\]
such that Eq. (\ref{ODE}) is verified and therefore%
\[
V=\phi^{\alpha}=\left(  t^{-\alpha b}\right)  =t^{-2b-2}\qquad
\Longleftrightarrow\qquad\alpha=\frac{2}{b}\left(  b+1\right)  .
\]

\hrulefill

The main quantities behave as follows%
\begin{equation}
\phi=\phi_{0}t^{-b},\qquad G=G_{0}t^{2b},\qquad V=V_{0}t^{-2\left(
b+1\right)  },\qquad H=ht^{-1},\qquad h\in\mathbb{R}.
\end{equation}

\hrulefill

In the same way we also may study the following equation
\begin{equation}
\phi^{\prime\prime}=-\frac{dV(\phi)}{d\phi}-\left(  \frac{2}{t}+\frac
{G^{\prime}}{G}\right)  \left(  \frac{1}{2}\phi^{\prime}+\frac{V(\phi)}%
{\phi^{\prime}}\right)  , \label{Te1}%
\end{equation}
through the Lie group method. Eq. (\ref{Te1}) admits the following symmetries%
\begin{align*}
\xi_{\phi\phi}  &  =0,\\
2\left(  \frac{2}{t}+\frac{G^{\prime}}{G}\right)  \xi_{\phi}+2\eta_{\phi\phi
}-4\xi_{t\phi}  &  =0,\\
6V_{\phi}\xi_{\phi}+\left(  \frac{2}{t}+\frac{G^{\prime}}{G}\right)  \xi
_{t}+\left(  -\frac{2}{t^{2}}+\frac{G^{\prime\prime}}{G}-\frac{G^{\prime2}%
}{G^{2}}\right)  \xi+4\eta_{t\phi}-2\xi_{tt}  &  =0,\\
8\left(  \frac{2}{t}+\frac{G^{\prime}}{G}\right)  V\xi_{\phi}+4V_{\phi}\xi
_{t}-2V_{\phi}\eta_{\phi}+\left(  \frac{2}{t}+\frac{G^{\prime}}{G}\right)
\eta_{t}+2V_{\phi\phi}\eta+2\eta_{tt}  &  =0,\\
\left(  \frac{2}{t}+\frac{G^{\prime}}{G}\right)  V_{\phi}\eta+3\left(
\frac{2}{t}+\frac{G^{\prime}}{G}\right)  V\xi_{t}-2\left(  \frac{2}{t}%
+\frac{G^{\prime}}{G}\right)  V\eta_{\phi}-\left(  -\frac{2}{t^{2}}%
+\frac{G^{\prime\prime}}{G}-\frac{G^{\prime2}}{G^{2}}\right)  V\xi &  =0,\\
\left(  \frac{1}{t}+\frac{G^{\prime}}{2G}\right)  V\eta_{t}  &  =0.
\end{align*}

As above, the symmetry
\begin{equation}
\xi=t,\qquad\eta=-b\phi\qquad\Longrightarrow\qquad\phi=\phi_{0}t^{-b},
\end{equation}
brings us to obtain the following constrain on function $G(t):$%
\begin{align}
\left(  \frac{2}{t}+\frac{G^{\prime}}{G}\right)  +\left(  -\frac{2}{t^{2}%
}+\frac{G^{\prime\prime}}{G}-\frac{G^{\prime2}}{G^{2}}\right)  t  &
=0,\label{aya1}\\
4V_{\phi}+2bV_{\phi}-2bV_{\phi\phi}\phi &  =0,\label{aya2}\\
-b\left(  \frac{2}{t}+\frac{G^{\prime}}{G}\right)  V_{\phi}\phi+3\left(
\frac{2}{t}+\frac{G^{\prime}}{G}\right)  V+2b\left(  \frac{2}{t}%
+\frac{G^{\prime}}{G}\right)  V-t\left(  -\frac{2}{t^{2}}+\frac{G^{\prime
\prime}}{G}-\frac{G^{\prime2}}{G^{2}}\right)  V  &  =0. \label{aya3}%
\end{align}

Eq. (\ref{aya3}) may be rewritten as%
\[
-b\left(  \frac{2}{t}+\frac{G^{\prime}}{G}\right)  V_{\phi}\phi+\left(
\left(  2b+3\right)  \left(  \frac{2}{t}+\frac{G^{\prime}}{G}\right)
-t\left(  -\frac{2}{t^{2}}+\frac{G^{\prime\prime}}{G}-\frac{G^{\prime2}}%
{G^{2}}\right)  \right)  V=0,
\]
while from Eq. (\ref{aya1}) we get%
\begin{equation}
G^{\prime\prime}=\frac{G^{\prime2}}{G}-\frac{G^{\prime}}{t},\qquad
\Longrightarrow\qquad G=G_{0}t^{k},\qquad k\in\mathbb{R}.
\end{equation}

From Eq. (\ref{aya2}) we get%
\begin{equation}
V_{\phi\phi}=\frac{2+b}{b}\frac{V_{\phi}}{\phi}\qquad\Longrightarrow\qquad
V=V_{0}\phi^{\frac{2}{b}\left(  b+1\right)  }=V_{0}t^{-2\left(  b+1\right)  }.
\end{equation}

Notice that we have obtained the same results as in the splitting case.

\hrulefill

The main quantities behave as follows%
\begin{equation}
\phi=\phi_{0}t^{-b},\qquad G=G_{0}t^{2b},\qquad V=V_{0}t^{-2\left(
b+1\right)  },\qquad H=ht^{-1},\qquad h\in\mathbb{R}.
\end{equation}

\hrulefill


\begin{thebibliography}{99}                                                                                               %


\bibitem {bernadis}P. Bernadis et al, Nature \textbf{404}, 955 (2000); S.
Hanany et al, Astrophys. J. Lett. \textbf{545}, L5 (2000); A. Balbi et al.,
Astrophys. J. Lett. \textbf{545}, L1-L4 (2000).

\bibitem {perlmutter}S. Perlmutter et al, Nature \textbf{391}, 51 (1998); S.
Perlmutter et al., Astrophys. J. \textbf{517}, 565 (1999); A. Riess et al,
Astron. J. \textbf{116}, 1009 (1998); P. M. Garnavich et al., Astrophys. J.
Lett. \textbf{493}, L53 (1998).

\bibitem {cc1}V. Sahni and A. Starobinsky, Int. J. Mod. Phys. D\textbf{9}, 373 (2000).

\bibitem {cc2}P. J. E. Peebles, Rev. Mod. Phys. \textbf{75}, 559 (2003).

\bibitem {cc3}T. Padmanabhan, Phys. Rep. \textbf{380}, 235 (2003).

\bibitem {cc4}T. Padmanabhan, Gen.Rel.Grav.\textbf{40}, 529-564, (2008).

\bibitem {Mia}J. M. F. Maia, J. A. S. Lima, Phys.Rev. \textbf{D65}, 083513 (2002).

\bibitem {CW}W. Chen, and Y-S, Wu\textbf{.} Phys. Rev. \textbf{D41}, 695,(1990).

\bibitem {quintessence}P. J. Peebles and B. Ratra, Astrophys. J. \textbf{325},
L17 (1988); B. Ratra and P. J. Peebles, Phys.\ Rev.\ D \textbf{37}, 3406
(1988); R. R. Caldwell, R. Dave and P. J. Steinhardt,
Phys.\ Rev.\ Lett.\ \textbf{80}, 1582 (1998); J. A. Frieman and I. Waga,
Phys.\ Rev.\ D \textbf{57}, 4642 (1998); I. Zlatev, L. M. Wang and P. J.
Steinhardt, Phys.\ Rev.\ Lett.\ \textbf{82}, 896 (1999).

\bibitem {steinhardt}C. Wetterich, Nucl. Phys. B\textbf{302}, 668 (1988); P.J.
Steinhardt, L. Wang and I. Zlatev, Phys. Rev. D\textbf{59}, 123504 (1999).

\bibitem {Pimen}L. M. Diaz-Rivera and L. O. Pimentel. Int.J.Mod.Phys.
\textbf{A18}, 651-672 (2003).

\bibitem {JBD1}C. Brans and R. H. Dicke, Phys. Rev. \textbf{124}, 925 (1961)

\bibitem {JBD2}P. G. Bergmann, Int. J. Theor. Phys. \textbf{1}, 25 (1968).

\bibitem {JBD3}K. Nordtvedt, Astrophys. J. \textbf{161}, 1059 (1970).

\bibitem {Fujii-Maeda}Y. Fujii and K- Maeda. \textquotedblleft The
scalar-Tensor Theory of gravitation\textquotedblright. CUP 2003.

\bibitem {Faraoni}V. Faraoni. \textquotedblleft Cosmology in Scalar-Tensor
Gravity\textquotedblright. Springer (31 Mar 2004).

\bibitem {bertolami}O. Bertolami and P.J. Martins, Phys. Rev. \textbf{D61},
064007 (2000).

\bibitem {sen}S. Sen and A.A. Sen, Phys. Rev. D\textbf{63}, 124006 (2001);
A.A. Sen and S. Sen, Mod. Phys. Lett A\textbf{16}, 1303 (2001); A.A. Sen, S.
Sen and S. Sethi, Phys. Rev. D\textbf{63},107501 (2001).

\bibitem {sen-ses}S. Sen and T. Seshadri, Int.J.Mod.Phys. D\textbf{12},
445-460 (2003).

\bibitem {diego}N. Banerjee and D. Pavon, Phys. Rev. D\textbf{63}, 043504 (2001).

\bibitem {bartolo}N. Bartolo and M. Pietroni, Phys. Rev. D\textbf{61}, 023518 (2000).

\bibitem {Chame}M.R. Setare et al: ArXiv:1006.0658.

\bibitem {ColeyDS}A.A. Coley, \textquotedblleft Dynamical Systems and
Cosmology\textquotedblright. Kluwer Academic Publishers (2003).

\bibitem {Tony1}J.A. Belinch\'{o}n, Class. Quantum Grav. \textbf{26,} 175003 (2009).

\bibitem {WE}J. Wainwright and G.F.R. Ellis: \textquotedblleft Dynamical
Systems in Cosmology\textquotedblright. Cambridge University Press (1997).



\bibitem {ES}H. Stephani, D. Kramer, M. MacCallum, et al., \textquotedblleft
Exact Solutions of Einstein's Field Equations\textquotedblright\ (2nd edn.).
Cambridge: Cambridge University Press (2003).

\bibitem {KB}K.A Bronnikov et al., Class. Quantum Grav. \textbf{21}, 3389-3403 (2004).

\bibitem {H}T. Harko and M.K. Mak, Int. J. Mod. Phys.D \textbf{11,} 1171 (2002).

\bibitem {Caminati}J. Caminati and R.G. Mclenaghan, J. Math. Phys.
\textbf{32}, 3135, (1991).

\bibitem {gron1}\O . Rudjord and $\;$\O . Gr\o n. Phys.Scripta. \textbf{77}, 055901,(2008).

\bibitem {gron2}\O . Gr\o n and S. Hervik. gr-qc/0205026.

\bibitem {Barrow}J.D. Barrow and S. Hervik. Class. Quant. Grav. \textbf{19},
5173, (2002).

\bibitem {Coley}W.C. Lim, A.A. Coley, S. Hervik. Class. Quant. Grav.
\textbf{24}, 595, (2006).

\bibitem {W}J. Wainwright, M.J. Hancock and C. Uggla. Class. Quantum Grav.
\textbf{16,} 2577 (1999).

\bibitem {penrose}R. Penrose, in General Relativity, an Einstein centenary
survey, eds. S.W. Hawking and W. Israel, Cambridge Univ. Press (1979).

\bibitem {wa}J. Wainwright and P.J. Anderson.Gen. Rel. Grav. \textbf{16}, 609 (1984).

\bibitem {ra}T. Rothman and P. Anninos, Phys. Lett. A\textbf{224}, 227 (1997).

\bibitem {rothman}T. Rothman, Gen. Rel. Grav. \textbf{32}, 1185 (2000).

\bibitem {Wain}J. Wainwright et al, Class. Quant. Grav. \textbf{16}, 2577, (2004).

\bibitem {lake}N. Pelavas and K. Lake, Phys. Rev. D\textbf{62}, 044009, (2000).

\bibitem {PC}N. Pelavas and A. Coley, Int.J.Theor.Phys. \textbf{45}, 1258-1266 (2006).



\bibitem {CC}B. J. Carr and A. A. Coley, Class. Quantum Grav. \textbf{16}, R31 (1999).

\bibitem {Hall}G. S. Hall. \textquotedblleft Symmetries and Curvature in
General Relativity\textquotedblright. World Scientific Lecture Notes in
Physics. Vol 46 (2004).

\bibitem {W2}J. Wainwright, Gen. Rel. Grav. \textbf{16}, 657 (1984).

\bibitem {Jantzen}K. Rosquits and R. Jantzen, Class. Quantum Grav. \textbf{2},
L129, (1985). K. Rosquits and R. Jantzen, \textquotedblleft Transitively
Self-Similarity Space-Times\textquotedblright. Proc. Marcel Grossmann Meeting
on General Relativity. Ed. Ruffini. Elsevier S.P. (1986). pg 1033.

\bibitem {HW}L.Hsu and J. Wainwright. Class. Quantum Grav, \textbf{3}, 1105-24,(1986).

\bibitem {Griego}P.S. Apostolopoulos. Class. Quantum Grav, \textbf{20}, 71-8433,(2003).



\bibitem {Wetterich}C. Wetterich. Astron. Astrophys \textbf{301}, 321-328 (1995).

\bibitem {Coleys}A.P Billyard, and A.A. Coley. Physical Review D\textbf{61},
083503 (2000).



\bibitem {Ibra}N. H. Ibragimov, \textquotedblleft Elementary Lie Group
Analysis and Ordinary Differential Equations\textquotedblright. Jonh Wiley \&
Sons, (1999).

\bibitem {olver}P. T. Olver, \textquotedblleft Applications of Lie Groups to
Differential Equations\textquotedblright\ . Springer-Verlang, (1993).

\bibitem {cantwell}B. J. Cantwell, \textquotedblleft Introduction to Symmetry
Analysis\textquotedblright. Cambridge University Press, Cambridge, (2002).

\bibitem {bluman}G.W Bluman and S.C. Anco. \textquotedblleft Symmetry and
Integral Methods for Differential Equations\textquotedblright.
Springer-Verlang (2002).



\bibitem {Tony2}J.A. Belinch\'{o}n. Gra\&Cos.\textbf{15}, 306-316, (2009).

\bibitem {Green88}M. B. Green et al. \textquotedblleft Superstring
theory\textquotedblright. CUP (1988).

\bibitem {Guzman}F.S. Guzman et al, Rev.Mex.Astron.Astrofis. \textbf{37},
63-72 (2001).

\bibitem {Huterer}D. Huterer and M.S. Turner, Phys. Rev. D\textbf{60}, 081301 (1999).

\bibitem {Tony3}J.A. Belinch\'{o}n, Int. Jour. Moder. Phys A\textbf{23},
5021-5036 (2008).



\bibitem {Will}C.M. Will. \textquotedblleft Theory and experiments in
gravitational physics\textquotedblright. CUP (1993).

\bibitem {46}R. D. Reasenberg et al., Astrophys. J., Lett. Ed. \textbf{234},
L219 (1979).

\bibitem {47}A. Liddle, A. Mazumdar, and J. Barrow, Phys. Rev. D \textbf{58},
027302 (1998).

\bibitem {48}M. Susperregi and A. Mazumdar, Phys. Rev. D \textbf{58}, 083512 (1998).

\bibitem {Yoshida1}H.\ Yoshida. Celestial mechanics, 31, \textbf{363}, (1983).

\bibitem {Yoshida2}H.\ Yoshida. Celestial mechanics, 31, \textbf{381}, (1983).



\bibitem {Sharif}M. Sharif, Int. J. Mod. Phys. D\textbf{14}, 1675-1684, (2005).

\bibitem {ccv}Carot, J., da Costa, J. and Vaz, E.G.L.R., J. Math. Phys.
\textbf{35}, 4832,(1994).

\bibitem {hrv}Hall, G.S., Roy, I. and Vaz, L.R.: Gen. Rel and Grav.
\textbf{28}, 299,(1996).

\bibitem {cc}Carot, J. and da Costa, J.: Procs. of the 6th Canadian Conf. on
General Relativity and Relativistic Astrophysics, Fields Inst. Commun. 15,
Amer. Math. Soc. WC Providence, RI(1997)179.

\bibitem {YC}Yavuz, \.{I}., and Camc\i, U.: Gen. Rel. Grav. \textbf{28}, 691,(1996).

\bibitem {YCB}Camc\i, U., Yavuz, \.{I}., Baysal, H., Tarhan, \.{I}., and Y\i
lmaz, \.{I}, Int. J. Mod. Phys. D\textbf{10}, 751,(2001).

\bibitem {YB}Camc\i, U. and Barnes, A.: Class. Quant. Grav. \textbf{19}, 393, (2002).

\bibitem {Sharif1}Sharif, M.: Nuovo Cimento B\textbf{116}, 673, (2001).

Sharif, M. Astrophys. Space Sci. \textbf{278}, 447, (2001).

\bibitem {TA}M. Tsamparlis and P.S. Apostolopoulos, Gen. Rel. and Grav.
\textbf{36}, 47, (2004).
\end{thebibliography}
\end{document}